\documentclass[acmsmall,screen]{acmart}
\makeatletter
\let\@authorsaddresses\@empty
\makeatother

\setcopyright{rightsretained}
\acmPrice{}
\acmDOI{10.1145/3473570}
\acmYear{2021}
\copyrightyear{2021}
\acmSubmissionID{icfp21main-p20-p}
\acmJournal{PACMPL}
\acmVolume{5}
\acmNumber{ICFP}
\acmArticle{65}
\acmMonth{8}

\usepackage{algorithm}
\usepackage[noend]{algpseudocode}
\usepackage{esvect}

\usepackage{ifthen} %

\usepackage{stmaryrd}
\usepackage{amsmath}
\usepackage{pgf}
\usepackage{tikz}
\usetikzlibrary{arrows,automata,positioning, shapes, shapes.geometric, fit, calc}
\usepackage[latin1]{inputenc}
\usepackage{picins}

\usepackage{proof}
\usepackage{enumitem}

\usepackage{booktabs}   %
\usepackage{subcaption} %
\usepackage{xspace}
\usepackage{soul}
\newcommand{\omittext}[1]{\ignorespaces}

\usepackage{upquote}
\usepackage{listings}

\usepackage[frozencache=true,cachedir=.]{minted}
\usepackage{etoolbox}
\makeatletter
\AtBeginEnvironment{minted}{\dontdofcolorbox}
\def\dontdofcolorbox{\renewcommand\fcolorbox[4][]{##4}}
\makeatother

\setlist{noitemsep,leftmargin=10pt,topsep=2pt,parsep=2pt,partopsep=2pt}

\newcommand{\binopdef}     \oplus %
\newcommand{\unopdef}      \ominus %

\makeatletter
\DeclareRobustCommand*\cal{\@fontswitch\relax\mathcal}
\makeatother

\def\cC{{\cal C}}

\def\cE{{\cal E}}

\def\cI{{\cal I}}

\def\cO{{\cal O}}

\def\cS{{\cal S}}

\newcommand{\tup} [1] {\langle #1 \rangle}

    \newcommand{\infral} [3] {\infer[\textsc{#3}]{\begin{array}{c} #2 \end{array} }{ \begin{array}{c} #1  \end{array} } }

\makeatletter
\renewcommand\section{\@startsection{section}{1}{\z@}%
                       {-8\p@ \@plus -4\p@ \@minus -4\p@}%
                       {6\p@ \@plus 4\p@ \@minus 4\p@}%
                       {\normalfont\large\bfseries\boldmath
                        \rightskip=\z@ \@plus 8em\pretolerance=10000 }}
\renewcommand\subsection{\@startsection{subsection}{2}{\z@}%
                       {-8\p@ \@plus -4\p@ \@minus -4\p@}%
                       {6\p@ \@plus 4\p@ \@minus 4\p@}%
                       {\normalfont\normalsize\bfseries\boldmath
                        \rightskip=\z@ \@plus 8em\pretolerance=10000 }}
\renewcommand\subsubsection{\@startsection{subsubsection}{3}{\z@}%
                       {-4\p@ \@plus -4\p@ \@minus -4\p@}%
                       {-1.5em \@plus -0.22em \@minus -0.1em}%
                       {\normalfont\normalsize\bfseries\boldmath}}
\makeatother

\newcommand{\eg}{{\em e.g.}, }
\newcommand{\ie}{{\em i.e.}, }
\newcommand{\etc}{{\em etc.}\xspace}

\newcommand{\kstar}{^{\textstyle *}}

\newcommand{\concat} {\mathsf{cat}}
\newcommand{\splitc} {\mathsf{split}}
\newcommand{\relay} {\mathsf{relay}}
\newcommand{\exec} [1] {\llbracket #1 \rrbracket}

\newcommand{\enablecomments}{true}

\ifthenelse{\equal{\enablecomments}{true}}
{\newcommand{\TODO}[1]{\hl{\textbf{TODO:} #1}\xspace}

\newif\ifshowcomment
\showcommentfalse
\showcommenttrue  %
\definecolor{revclr}{HTML}{004775}
\newcommand{\nv}[1]{\ifshowcomment [{\color{cyan}Nikos:  #1}] \else \fi}
\newcommand{\km}[1]{\ifshowcomment [{\color{blue}Shivam: #1}] \else \fi}
\newcommand{\kk}[1]{\ifshowcomment [{\color{magenta}Kon: #1}] \else \fi}

\newcommand{\rev}[1]{#1}

}
{\newcommand{\TODO}[1]{}

\newcommand{\nv}[1]{}
\newcommand{\km}[1]{}
\newcommand{\kk}[1]{}

\newcommand{\rev}[1]{#1}

}

\newcommand{\eat}[1]{}
\newcommand{\tr}[1]{}

\newcommand{\sx}[1]{(\S\ref{#1})}

\newcommand{\unix}{{\scshape Unix}\xspace}

\newcommand{\heading}[1]{\vspace{4pt}\noindent\textbf{#1}:\enspace}

\newcommand{\ttt}[1]{\mintinline[fontsize=\small]{bash}{#1}}
\newcommand{\ttiny}[1]{\mintinline[fontsize=\footnotesize]{bash}{#1}}

\lstdefinelanguage{sh}{
  morekeywords={for, in, do, done, \|},
  keywordstyle=\color{purple}\ttfamily,
  ndkeywordstyle=\color{black}\ttfamily\bfseries,
  identifierstyle=\color{black},
  sensitive=false,
  comment=[l]{\#},
  commentstyle=\color{lightgray},
  stringstyle=\color{darkgray}\ttfamily,
  morestring=[b]',
  morestring=[b]",
  abovecaptionskip=0pt,
  aboveskip=0pt,
  belowcaptionskip=0pt,
  belowskip=0pt,
  frame=none                     %
}

\begin{document}

\title{An Order-Aware Dataflow Model for Parallel Unix Pipelines}

\author{Shivam Handa}
\authornote{Equal contribution.}
\affiliation{
  \institution{CSAIL, MIT}       %
  \country{USA}                       %
}
\email{shivam@csail.mit.edu}          %

\author{Konstantinos Kallas}
\authornotemark[1]
\affiliation{
  \institution{University of Pennsylvania}       %
  \country{USA}                                  %
}
\email{kallas@seas.upenn.edu}                    %

\author{Nikos Vasilakis}
\authornotemark[1]
\affiliation{
  \institution{CSAIL, MIT}            %
  \country{USA}                       %
}
\email{nikos@vasilak.is}              %

\author{Martin C. Rinard}
\affiliation{
  \institution{CSAIL, MIT}            %
  \country{USA}                       %
}
\email{rinard@csail.mit.edu}          %

\begin{abstract}
    We present a dataflow model for modelling parallel \unix shell pipelines.
To accurately capture the semantics of complex \unix pipelines, the dataflow model is order-aware, \ie the order in which a node in the dataflow graph consumes inputs from different edges plays a central role in the semantics of the computation and therefore in the resulting parallelization.
We use this model to capture the semantics of transformations that exploit data parallelism available in \unix shell computations and prove their correctness.
We additionally formalize the translations from the \unix shell to the dataflow model and from the dataflow model back to a parallel shell script.
We implement our model and transformations as the compiler and optimization passes of a system parallelizing shell pipelines, and use it to evaluate the speedup achieved on 47 pipelines.

\end{abstract}



\maketitle

\section{Introduction}
\label{intro}

\unix pipelines are an attractive choice for specifying succinct and simple programs for data processing, system orchestration, and other automation tasks~\cite{mcilroy1978unix}.
Consider, for example, the following program based on the original \ttt{spell} written by Johnson~\cite{bentley1985spelling}, lightly modified for modern environments:\footnote{
  Johnson's program additionally used \ttiny{troff}, \ttiny{prepare}, and \ttiny{col -bx} to clean up now-legacy formatting metadata that does not exist in markdown.
  Moreover, \ttiny{comm -13} was replaced with \ttiny{grep -xvf} to highlight crucial features of the model.
}

\medskip
\begin{minted}[fontsize=\footnotesize]{bash}
cat f1.md f2.md | tr A-Z a-z | tr -cs A-Za-z '\n' | sort | uniq |                   # (Spell)
  grep -vx -f dict.txt - > out ; cat out | wc -l | sed 's/$/ mispelled words!/'
\end{minted}
\medskip

\noindent
The first command streams two markdown files into a pipeline that converts characters in the stream into lower case, removes punctuation, sorts the stream in alphabetical order, removes duplicate words, and filters out words from a dictionary file (lines 1 and 2, up to ``\ttt{;}'').
A second pipeline (line 2, after ``\ttt{;}'') counts the resulting lines to report the number of misspelled words to the user.

As this example illustrates, the \unix shell offers a programming model that facilitates the composition of commands using unidirectional communication channels that feed the output of one command as an input to another. 
These channels are either ephemeral, unnamed pipes expressed using the \ttt{|} character and lasting for the duration of the producer and consumer, or persistent, named pipes (\unix FIFOs) created with \ttt{mkfifo} and lasting until explicitly deleted.
Each command executes sequentially, with pipelined parallelism available between commands executing in the same pipeline.
Unfortunately, this model leaves substantial data parallelism, \ie parallelism achieved by splitting inputs into pieces and feeding the pieces to parallel instances of the script, unexploited.

This fact is known in the \unix community and has motivated the development of a variety of tools that attempt to exploit latent data parallelism in \unix scripts~\cite{Tange2011a, pash, posh}.
On the one hand, tools such as GNU Parallel~\cite{Tange2011a} can be used by experienced users to achieve parallelism, but could also easily lead to incorrect results.
On the other hand, two recent systems, PaSh~\cite{pash} and POSH~\cite{posh}, focus respectively on extracting data parallelism latent in \unix pipelines to improve the execution time of
  (i) CPU-intensive shell scripts (PaSh), and (ii) networked IO-intensive shell scripts (POSH).
These systems achieve order-of-magnitude performance improvements on sequential pipelines, but their semantics and associated transformations are not clearly defined, making it difficult to ensure that the optimized parallel scripts are sound with respect to the sequential ones.

To support the ability to reason about and correctly transform \unix shell pipelines, we present a new {\em dataflow model}.
In contrast to standard dataflow models~\cite{kahn1974semantics, kahn1976coroutines, lee1987synchronous, lee1987static, karp1966properties}, our dataflow model is {\em order-aware}---\ie the order in which a node in the dataflow graph consumes inputs from different edges plays a central role in the semantics of the computation and therefore in the resulting parallelization. 
This model is different from models that allow multiplexing different chunks of data in a single channel, such as sharding or tagging, or ones that are oblivious to ordering, such as shuffling---and is a direct byproduct of the ordered semantics of the shell and the opacity of \unix commands.
In the \emph{Spell} script shown earlier, for example, while all commands consume elements from an input stream in order---a property of \unix streams \eg pipes and FIFOs---they differ in how they consume across streams:
  \ttt{cat} reads input streams in the order of its arguments,
  \ttt{sort -m} reads input streams in interleaved fashion, and
  \ttt{grep -vx -f} first reads \ttt{dict.txt} before reading from its standard input.
  
We use this order-aware dataflow model (ODFM) to express the semantics of transformations that exploit data parallelism available in \unix shell computations.
These transformations capture the parallelizing optimizations performed by both PaSh~\cite{pash} and POSH~\cite{posh}.
We also use our model to prove that these transformations are correct, \ie that they do not affect the program behavior with respect to sequential output.
Finally, we formalize bidirectional translations between the shell and ODFM, namely from the shell language to ODFM and \emph{vice versa}, closing the loop for a complete shell-to-shell parallelizing compiler.

To illustrate the applicability of our model, we extend PaSh by reimplementing its compilation and optimization phases with ODFM as the centerpiece.
The new implementation translates fragments of a shell script to our dataflow model, applies a series of parallelizing transformations, and then translates the resulting parallel dataflow graph back to a parallel shell script that is executed instead of the original fragment.
Our new implementation improves modularity and facilitates the development of different transformations independently on top of ODFM.
We use the new implementation to evaluate the benefit of a specific transformation by parallelizing 47 unmodified shell scripts with and without this transformation and measuring their execution times.
Finally, we present a case study in which we parallelize two scripts using GNU Parallel~\cite{Tange2011a};
our experience indicates that while it is easy to parallelize shell scripts using such a tool, it is also easy to introduce bugs, leading to incorrect results.

In summary, this paper makes the following contributions:
\begin{itemize}
  \item {\bf Order-Aware Dataflow Model:} It introduces the \emph{order-aware dataflow model} (ODFM), a dataflow model tailored to the \unix shell that captures information about the order in which nodes consume inputs from different input edges~\sx{model}.
  \item {\bf Translations:} It formalizes the bidirectional translations between shell and ODFM required to translate shell scripts into dataflow graphs and dataflow graphs back to parallel scripts~\sx{translations}. %
  \item {\bf Transformations and Proofs of Correctness:} It presents a series of ODFM transformations for extracting data parallelism. It also presents proofs of correctness for these transformations~\sx{transformations}.
  \item {\bf Results:} It reimplements the optimizing compiler of PaSh and presents experimental results that evaluate the speedup afforded by a specific dataflow transformation~\sx{eval}.
\end{itemize}

\noindent
The paper starts with an informal development building the necessary background~\sx{bg} and expounding on \emph{Spell}~\sx{informal}.
It then presents the three main contributions outlined above (\S\ref{model}--\ref{eval}), compares with prior work~\sx{related}, \rev{and offers a discussion~\sx{discussion}}, before closing the paper~\sx{conclusion}.

\section{Background}
\label{bg}

This section reviews background on commands and abstractions in the \unix shell.

\heading{\unix Streams}
A key \unix abstraction is the data \emph{stream}, operated upon by executing commands or \emph{processes}.
Streams are sequences of bytes, but most commands process them as higher-level sequences of line elements, with the newline character delimiting each element and the {\sc EOF} condition representing the end of a stream.
Streams are often referenced using a filename, that is an identifier in a global name-space made available by the \unix file-system such as \ttt{/home/user/x}.
Streams and files in \unix are isomorphic:
  some streams can persist as files beyond the execution of the process, whereas other streams are ephemeral in that they only exist to connect the output of one process to the input of another process during their execution.
The sequence order is maintained when changing between persistent files and ephemeral streams.

\heading{Commands}
Each command is an independent computation unit that reads one or more input streams, performs a computation, and produces one or more output streams.
Contrary to languages with a closed set of primitives, there is an unlimited number of \unix commands, each one of which may have arbitrary behaviors---with the command's side-effects potentially affecting the entire environment on which it is executing.
These commands may be written in any language or exist only in binary form, and thus \unix is not easily amenable to a single parallelizability analysis.
Parallelization tools such as GNU \ttt{parallel} leave such analysis to developers that have to ensure that the script behavior will not be affected by parallelization,
  whereas transformation-based tools such as PaSh and POSH identify key invariants that hold for entire classes of commands and then resort to annotation libraries to infer whether each invariant is satisfied by each command.
For example, an invariant that is used in both PaSh and POSH is whether a command is stateless, \ie whether it maintains state when processing its inputs, or whether it processes each input line independently.
Commands that satisfy this invariant can be parallelized by splitting their inputs in lines and then combining their outputs.

\heading{Command flags}
\unix commands are often configurable, customizing their behavior based on the task at hand.
This is usually achieved via environment variables and command flags.
Tools like PaSh and POSH address command behavior variability due to flags by including flags and command arguments in their annotation frameworks. Given command annotations, PaSh and POSH can abstract specific command invocations in a pipeline to black boxes for which some assumptions hold. This makes them applicable in the context of the shell where the space of possible command and flag combinations is exceedingly large.

\heading{Order of input consumption}
In \unix, all streams are ordered and all commands can safely assume that they can consume elements from their streams in the order they were produced.
Additionally, most commands have the ability to operate on multiple files or streams.
The order in which commands access these streams is important.
In some cases, they read streams in the order of the stream identifiers provided.
In other cases, the order is different---for example, an input stream may configure a command, and thus must be read before all the others.
Consider for example \ttt{grep -f words.txt input.txt}, which first reads \ttt{words.txt} to determine the keywords for which it needs to search, and then reads \ttt{input.txt} line by line, emitting all lines that contain one of the words in \ttt{words.txt}.
In other cases, reads from multiple streams are interleaved according to some command-specific semantics.

\heading{Composition: \unix Operators}
\unix provides several primitives for program composition, each of which imposes different scheduling constraints on the program execution.
Central among them is the \emph{pipe} (\ttt{|}), a primitive that passes the output of one process as input to the next.
The two processes form a pipeline, producing output and consuming input concurrently and possibly at different rates.
The \unix kernel facilitates program scheduling, communication, and synchronization behind the scenes.
For example, \emph{Spell}'s first \ttt{tr} transforms each character in the input stream to lower case, passing the stream to the second \ttt{tr}:
  the two \ttt{tr}s form a parallel producer-consumer pair of processes.  

Apart from pipes, the language of the \unix shell provides several other forms of program composition---\eg 
  the sequential composition operator (\ttt{;}) for executing one process after another has completed, 
  and  control structures such as \ttt{if} and \ttt{while}.
All of these constructs %
  enforce execution ordering between their components.
To preserve such ordering and thus ensure correctness, systems such as PaSh and POSH do not ``push'' parallelization beyond these constructs.
Instead, they focus on exploiting parallelism in script regions that do not face ordering constraints---which, as they demonstrate, is enough to significantly improve the performance of scripts found out in the wild~\cite{pash,posh}.

\section{Example and Overview}
\label{informal}

This section provides intuition of the order-aware dataflow model proposed in this paper by following the different phases of a shell-to-shell parallelizing compiler (inspired by PaSh and POSH), formalized in the later sections. 
Given a script such as \emph{Spell}~\sx{intro}, the compiler identifies its dataflow regions, translates them to DFGs (Shell$\rightarrow$ODFM), applies graph transformations that expose data parallelism on these DFGs, and replaces the original dataflow regions with the now-parallel regions (ODFM$\rightarrow$Shell).

\pichskip{12pt}
\parpic[r][t]{
  \begin{minipage}{50mm}
    \includegraphics[width=\linewidth]{\detokenize{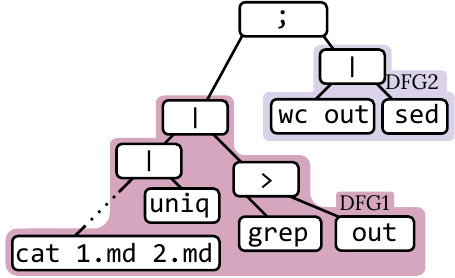}}
  \end{minipage}
}
\heading{Shell$\rightarrow$ODFM}
Provided a shell script, the compiler starts by identifying subexpressions that are potentially parallelizable.
The first step is to parse the script, creating an abstract syntax tree like the one presented on the right.
Here we %
  omit any non-stream flags and refer to all the stages between (and including) \ttt{tr} and \ttt{sort} as a dotted edge ending with \ttt{cat}.

The compiler then identifies parallelism barriers within the shell script:
  these barriers are operators that enforce synchronization constraints such as the sequential composition operator (``\ttt{;}''). 
We call any set of commands that does not include a dataflow barrier a {\em dataflow region}.
Dataflow regions are then transformed to dataflow graphs (DFGs), i.e., instances of our order-aware dataflow model.
In our example, there are two dataflow regions corresponding to the following dataflow graphs:
\begin{center}
\includegraphics[width=0.7\textwidth]{\detokenize{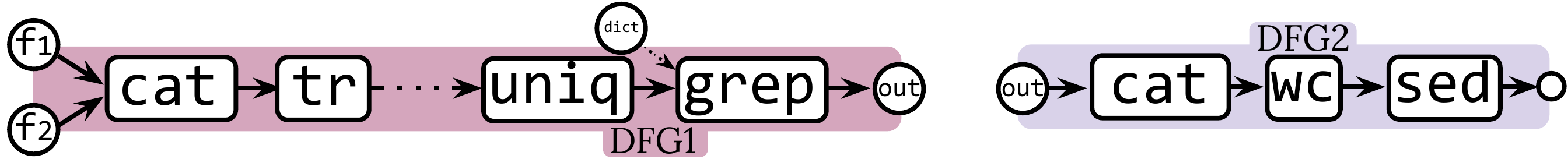}}
\end{center}

\noindent
As mentioned earlier \sx{bg}, the compiler exposes parallelism in each DFG separately to preserve the ordering requirements imposed to ensure correctness. %
For the rest of this section we focus on the parallelization of DFG1.

\parpic[r][b]{
  \begin{minipage}{45mm}
    \center
    \footnotesize
    \setlength\tabcolsep{2.5pt}
    \begin{tabular}{ll}
      \toprule
      Command                            & Aggreg. Function          \\
      \midrule
      \ttiny{cat}                          & \ttiny{cat $*}            \\
      \ttiny{tr A-Z a-z}                   & \ttiny{cat $*}            \\
      \ttiny{tr -d a}                      & \ttiny{cat $*}            \\
      \ttiny{sort}                         & \ttiny{sort -m $*}        \\
      \ttiny{uniq}                         & \ttiny{uniq $*}           \\ %
      \ttiny{grep -f a -}                  & \ttiny{cat $*}            \\
      \ttiny{wc -l}                        & \ttiny{paste -d+ $*|bc} \\
      \ttiny{sed 's/a/b/'}                 & \ttiny{cat $*}            \\
      \bottomrule
    \end{tabular}
  \end{minipage}
}
\heading{Parallelizable Commands}
Individual nodes of the dataflow graphs are shell commands.
Systems like PaSh and POSH assume key information for individual commands, \eg whether they are amenable to divide-and-conquer data parallelism.
Such data parallelism is achieved by splitting the input into pieces (at stream element boundaries), processing partial inputs in parallel, and finally applying an aggregation function to partial outputs to produce the final output.
This decomposition breaks a command into two components---a data-parallel function, which is often the command itself, and an aggregation function.
The table on the right presents aggregation functions for the shell commands in our example (all of which are parallelizable).

\pichskip{12pt}
\parpic[l][t]{
  \begin{minipage}{30mm}
    \includegraphics[width=1\linewidth]{\detokenize{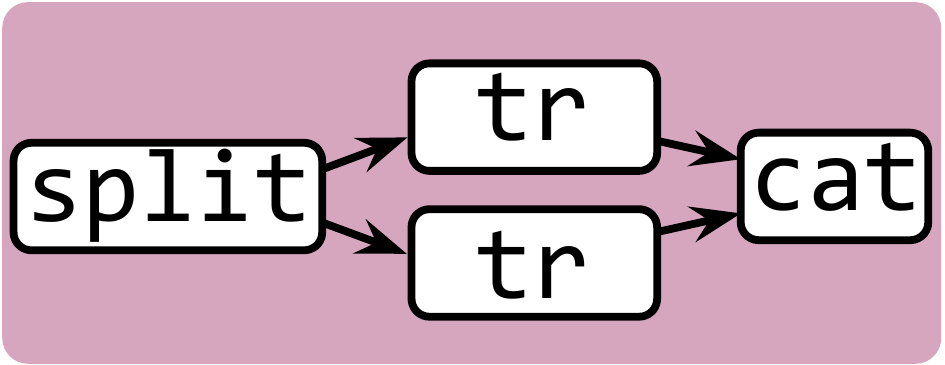}}
  \end{minipage}
}

\noindent
For example, consider the decomposition of the \ttt{tr} command.
Applying \ttt{tr} over the entire input produces the same result as \ttt{split}ting the input into two, applying \ttt{tr} to the two partial inputs, and then merging the partial results with a \ttt{cat} aggregation function.
Note that both \ttt{split} and \ttt{cat} are order-aware, \ie \ttt{split} sends the first half of its input to the first \ttt{tr} and the rest to the second, while \ttt{cat} concatenates its inputs in order.
This guarantees that the output of the DFG is the same as the one before the transformation.

\heading{Parallelization Transformations}
Given the decomposition of individual commands, the compiler's next step is to apply graph transformations to exploit parallelism present in the computation represented by the DFG. 
As each parallelizable \unix command comes with a corresponding aggregation function, the compiler's transformations first convert the DFG into one that exploits parallelism at each stage.
After applying the transformation to the two \ttt{tr} stages, the DFG looks as follows:
\begin{center}
\includegraphics[width=0.7\linewidth]{\detokenize{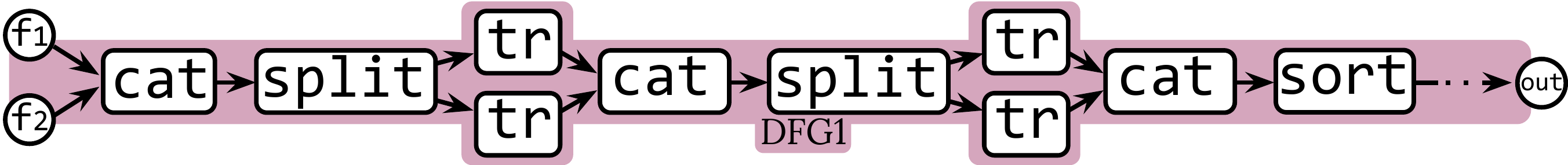}}
\end{center}
\noindent
After these transformations are applied to all DFG nodes, the next transformation pass is applied to pairs of \ttt{cat} and \ttt{split} nodes:
  whenever a \ttt{cat} is followed by a \ttt{split} of the same width, the transformation removes the pair and connects the parallel streams directly to each other.
The goal is to \emph{push} data parallelism transformations as far down the pipeline as possible to expose the maximal amount of parallelism. 
Here is the resulting DFG for the transformation applied to the two \ttt{tr} stages:
\begin{center}
\includegraphics[width=0.7\linewidth]{\detokenize{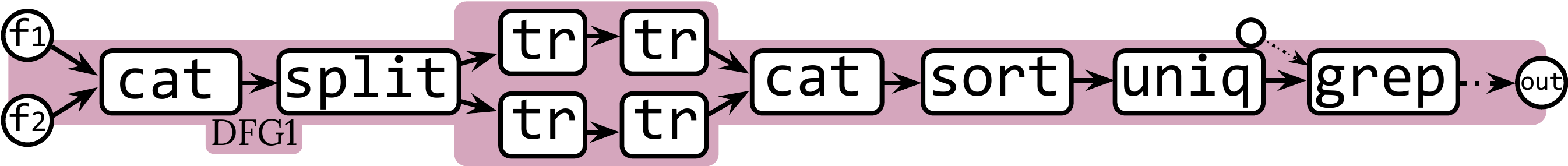}}
\end{center}
\noindent
Applying this transformation to the first three stages---\ie \ttt{cat}, \ttt{tr}, and \ttt{tr}---of DFG1 produces the following transformed DFG.
\begin{center}
\includegraphics[width=0.7\linewidth]{\detokenize{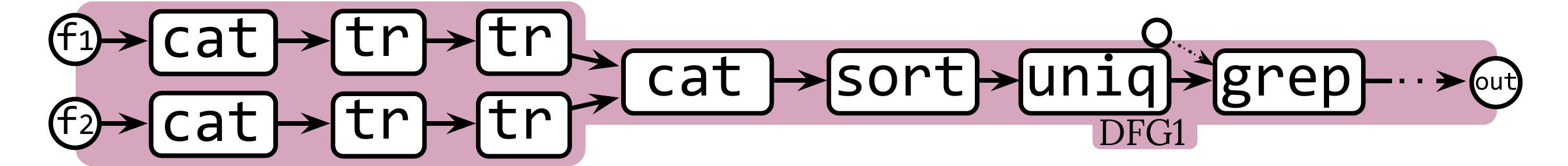}}
\end{center}
\noindent
The next node to parallelize is \ttt{sort}.
To merge the partial output of parallel \ttt{sort}s, we need to apply a sorted merge.
(In GNU systems, this is available as \ttt{sort -m} so we use this as the label of the merging node.)
The transformation then removes \ttt{cat}, replicates \ttt{sort}, and merges their outputs with \ttt{sort -m}:
\begin{center}
\includegraphics[width=0.7\linewidth]{\detokenize{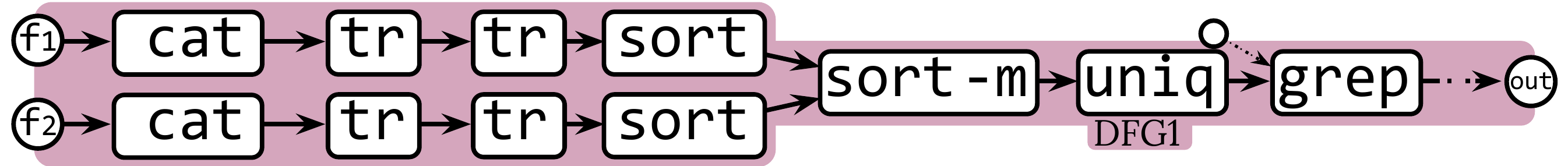}}
\end{center}
\noindent
It then continues pushing parallelism down the pipeline, after applying a \ttt{split} function to split \ttt{sort -m}'s outputs.
\begin{center}
\includegraphics[width=0.7\linewidth]{\detokenize{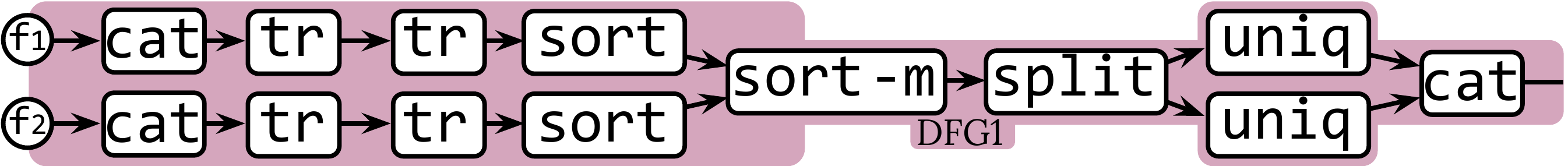}}
\end{center}

\noindent
As mentioned earlier, a similar pass of iterative transformations is applied to DFG2, but the two DFGs are not merged to preserve the synchronization constraint of the dataflow barrier ``\ttt{;}''.

\heading{Order Awareness}
Data-parallel systems~\cite{mapreduce:08,spark:12} often achieve parallelism using sharding, \ie partitioning input based on some key, or using shuffling, \ie arbitrary partitioning of inputs to parallel instances of an operator.

However, these techniques cannot be directly applied to the context of the shell, since (1) \unix commands and pipelines assume strict ordering of their input elements, (2) most commands are not independent on the basis of some key (to enable sharding), and (3) many commands are not commutative (\eg \ttt{uniq}, \ttt{cat -n}).
Since our goal is to define a model that applies directly to existing shell scripts, we cannot simply introduce new primitives that support sharding or shuffling, as is done in the case of systems that design an abstraction that fits their needs (\eg MapReduce, Spark).
Thus, data parallelism in the shell requires a careful treatment of input and output ordering.

\pichskip{12pt}
\parpic[r][t]{
  \begin{minipage}{15mm}
    \includegraphics[width=\linewidth]{\detokenize{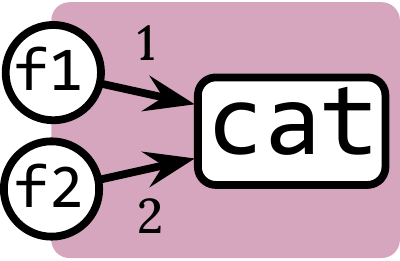}}
  \end{minipage}
}
To further explain the need for order-awareness in a model for data parallel \unix pipelines, let's look at the following examples.
Consider \emph{Spell}'s \ttt{cat f1.md f2.md} command that starts reading from \ttt{f2.md} only after it has completed reading \ttt{f1.md};
  note that any or both input streams may be pipes waiting for results from other processes.
This order can be visualized as a label over each input edge.
Correctly parallelizing this command requires ensuring that parallel \ttt{cat} (and possibly followup stages) maintains this order.

\pichskip{12pt}
\parpic[r][t]{
  \begin{minipage}{32mm}
    \includegraphics[width=\linewidth]{\detokenize{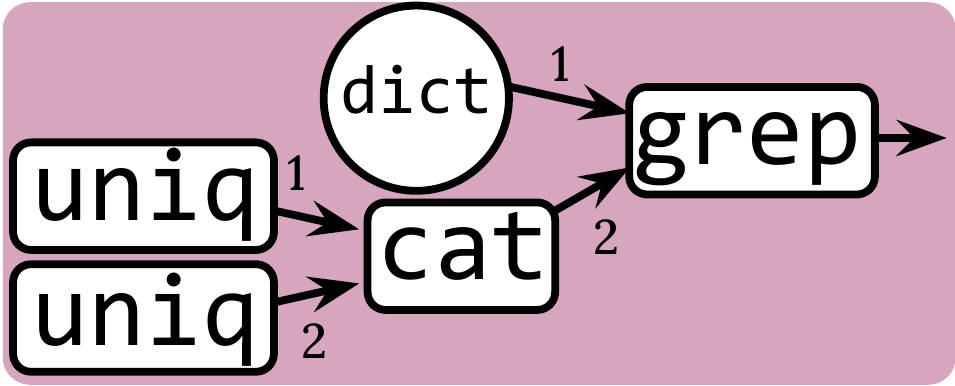}}
  \end{minipage}
}
As a more interesting example, consider \emph{Spell}'s \ttt{grep}, whose DFG is shown on the right.
Parallelizing \ttt{grep} without taking order into account is not trivial, because \ttt{grep -vx -f}'s set difference is not commutative:
  we cannot simply split its input streams into two pairs of partial inputs fed into two copies of \ttt{grep}.
Taking input ordering into account, however, highlights an important dependency between \ttt{grep}'s inputs.
The \ttt{dict} stream can be viewed as configuring \ttt{grep}, and thus \ttt{grep} can be modeled as consuming the entire \ttt{dict} stream before consuming partial inputs.

\pichskip{12pt}
\parpic[r][t]{
  \begin{minipage}{42mm}
    \includegraphics[width=\linewidth]{\detokenize{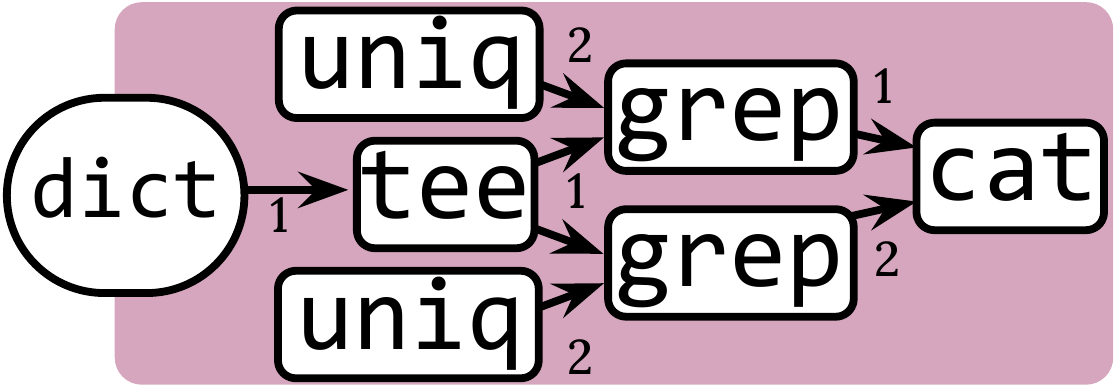}}
  \end{minipage}
}
Armed with this insight, the compiler parallelizes \ttt{grep} by passing the same \ttt{dict.txt} stream to both \ttt{grep} copies.
This requires an intermediary \ttt{tee} for duplicating the \ttt{dict.txt} stream to both copies of \ttt{grep}, each of which consumes the stream in its entirety before consuming the results of the preceeding \ttt{uniq}.

Order-awareness is also important for the DFG translation back to a shell script.
In the specific example we need to know how to instantiate the arguments of each \ttt{grep} of all possible options---\eg
  \ttt{grep -vx -f p1 p2}, \ttt{cat p1 | grep -vx -f - p2}, \etc
Aggregators are \unix commands with their own ordering characteristics that need to be accounted for.

The order of input consumption in the examples of this section is statically known and can be represented for each node as a set of configuration inputs, plus a sequence of the rest of its inputs.
To accurately capture the behavior of shell programs, however, ODFM is more expressive, allowing any order of input consumption.
The correctness of our parallelization transformations is predicated upon static but configurable orderings:
  a command reads a set of \emph{configuration} streams to setup the consumption order of its input streams which are then consumed in-order, one after the other.

\heading{ODFM$\rightarrow$Shell}
The transformed graph is finally compiled back to a script that uses POSIX shell primitives to drive parallelism explicitly.
A benefit of the dataflow model is that it can be directly implemented on top of the shell, simply translating each node to a command, and each edge to a stream.
The generated parallel script for \emph{Spell} can be seen below.

\medskip
\begin{minted}[fontsize=\footnotesize]{bash}
mkfifo t{0..14}         # DFG1: start
tr A-Z a-z < f1.md > t0 &
tr A-Z a-z < f2.md > t1 &
tr -d[:punct:] < t0 > t2 &
tr -d[:punct:] < t1 > t3 &
sort < t2 > t4 &
sort < t3 > t5 &
sort -m t4 t5 > t6 &
split t7 t8 < t6 &
# ...
tee t9 > t10 < dict.txt &
grep -vx -f t9 - < t11 > t13 &
grep -vx -f t10 - < t12 > t14 &
cat t13 t14 > out &
wait
rm t{0..14}               # DFG1: end

mkfifo t{0..8}          # DFG2: start
split t0 t1 < out &
wc -l < t0 > t2 &
wc -l < t1 > t3 &
paste -d+ t2 t3 | bc > t4 &
split t5 t6 < t4 &
sed 's/$/ mispelled words!/' < t5 > t7 &
sed 's/$/ mispelled words!/' < t6 > t8 &
cat t7 t8 &
wait
rm t{0..8}                # DFG2: end
\end{minted}
\medskip
The two DFGs are compiled into the two fragments that start with \ttt{mkfifo} and end with \ttt{rm}.
Each fragment uses a series of named pipes (FIFOs) to explicitely manipulate the input and output streams of each data-parallel instance, effectively laying out the structure of the DFG using explicit channel naming (\unix FIFOs are named in the filesystem similar to normal files.)
Aggregation functions are used to merge partial outputs from previous commands coming in through multiple FIFOs---for example, \ttt{sort -m t4 t6} and \ttt{cat t11 t12} for the first fragment, and \ttt{paste -d+ t2 t3 | bc} and \ttt{cat t7 t8} for the second.
A \ttt{wait} blocks until all commands executing in parallel complete.

The parallel script is simplified for clarity of exposition:
  it does not show the details of input splitting, handling of \ttt{SIGPIPE} deadlocks, and other technical details that are handled by the current implementation.

Readers might be wondering about the correctness of having two \ttt{sed} commands in the parallel script:
  won't the string ``mispelled words'' appear twice in the output?
Note, however, that the output of the \ttt{wc} stage (fifo \ttt{t4}) contains a single line.
As a result, the second \ttt{sed} will not be given any input line and thus will not produce any output.

\section{An Order-aware Dataflow Model}
\label{model}

In this section we describe the order-aware dataflow model (ODFM) and its semantics.

\subsection{Preliminaries}
\label{model:prelim}

As discussed earlier~\sx{bg}, the two main shell abstractions are (i) data streams, and (ii) commands communicating via streams.
We represent streams as named variables and commands as functions that read from and write to streams.

We first introduce some basic notation formalizing data streams on which our dataflow description language works. 
For a set $D$, we write $D\kstar$ to denote the set of all finite words over $D$.
For words $x, y \in D\kstar$, we write $x \cdot y$ or $xy$ to denote their concatenation.
We write $\epsilon$ for the empty word and $\bot$ for the End-of-File condition.
We say that $x$ is a \emph{prefix} of $y$, and we write $x \leq y$, 
if there is a word $z$ such that $y = xz$.
The $\leq$ order is reflexive, antisymmetric, and transitive 
(\ie it is a partial order), and is often called the \emph{prefix order}.
We use the notation $D\kstar \cdot \bot$ to denote a 
\emph{closed stream}, abstractly representing
a file/pipe stream that has been closed, \ie one which no process will open for
writing.
The notation $D\kstar$ is used to denote an \emph{open stream}, abstractly representing
an open pipe. Later, other process may add new elements at the end of this
value. 
\rev{
In the rest of out formalization we focus on terminating streams, and therefore terminating programs, since all of the data processing scripts that we have encountered are terminating.
We discuss extensions for infinite streams in \S\ref{discussion}.
}

\begin{figure}%
    \[
        \begin{array}{rcl}
            P &:=& \cI ; \cO ; \overline{\cE}\\
            \cI &:=& \mathsf{input}~\overline{x}\\
            \cO &:=& \mathsf{output}~\overline{x}\\
            \cE &:=& \overline{x}_o \leftarrow f(\overline{x}_i)\\
        \end{array}
    \] 
    \caption{Dataflow Description Language (DDL).
    }
    \label{fig:dataflow_descrip}
\end{figure}

\subsection{Dataflow Description Language}
\label{graph-components}

Figure~\ref{fig:dataflow_descrip} presents the Dataflow Description Language (DDL) for defining dataflow graphs (DFG).
A program $p \in P$ in DDL is of the form $\cI ; \cO ; \overline{\cE}$.
$\cI$ and $\cO$ represent sets of edges, vectors of the form $\overline{x} = \tup{x_1, x_2, \ldots x_n}$.
Variables $x_1, x_2, \ldots $ represent DFG edges, \ie streams used as a communication channel between DFG nodes and as the input and output of the entire DFG.

$\cI$ is of the form $\mathsf{input}~\overline{x}$, where $\overline{x}$ is the set of the input variables.
Each variable $x \in \cI$ represents a file $\texttt{file}(x)$ that is read from the \unix filesystem.
Note that multiple input variables can refer to the same file. 

$\cO$ is of the form $\mathsf{output}~\overline{x}$, where $\overline{x}$ is the set of output variables.
Each variable $x \in \cO$ represents a file $\texttt{file}(x)$ that is written to the \unix filesystem.

$\cE$ represents the nodes of the DFG.
A node $\overline{x}_o \leftarrow f(\overline{x}_i)$  represents a function from list of input variables (edges) $\overline{x}_i$ to output variables (edges) $\overline{x}_o$.  
We require that $f$ is monotone with respect to a lifting of the prefix order for a sequence of inputs;
  that is, $\forall, v, v', \overline{v}_i$, if $v \leq v'$, $ \tup{v_1, \ldots v_n} = f(v, \overline{v}_i)$ and $\tup{v'_1, \ldots v'_n} =  f(v', \overline{v}_i)$, then
$\forall~k \in [1, n].~v_k \leq v'_k$.
This captures the idea that a node cannot retract output that it has already produced. 

We wrap all functions $f$ with an execution wrapper $\exec{\cdot}$ that ensures that all outputs of $f$ are closed when its inputs are closed:
\[\exec{f(v_1 \cdot \bot, v_2 \cdot \bot, \ldots v_n \cdot \bot)} = \tup{v'_1 \cdot \bot, v'_2 \cdot \bot, \ldots v'_k \cdot \bot}\]
This is helpful to ensure termination.
From now on, we only refer to the wrapped function semantics.

We also assume that 
commands do not produce output if they have not consumed any input, \ie the following is true:
\[
    \tup{\epsilon, \ldots,\epsilon} = \exec{f(\epsilon, \ldots, \epsilon)}
\]

A variable in DDL is assigned \emph{only once} and consumed by \emph{only one} node. 
DDL does not allow the dataflow graph to contain \emph{any cycles}.
This also holds for variables in $\cI$ and $\cO$, which cannot refer to the same variables in $\cI$ and never assigned a different value in $\cE$. 
Similarly, variables in $\cO$ are not read by any node in $\cE$.
All variables which are not included in $\cI$ and $\cO$ abstractly represent temporary files/pipes which are created during the execution of a shell script.
We assume that within a dataflow program, all variables are reachable from some input variables.

\begin{figure}
    \begin{center}$    
    \begin{array}{c}
        \infral{
                \begin{array}{ccc}
                    \tup{x'_1, \ldots x'_p} \leftarrow f(x_1, \ldots x_k \ldots
                    x_n) \in \cE &
                    v_k \cdot v_x \leq \Gamma(x_k) \\
                    \vert v_x \vert = 1 \vee v_x = \bot &
                    i \in [1, k-1] \cup [k+1, n]. v_i = \sigma(x_i) \\
                    k \in \mathsf{choice}_f(v_1, \ldots v_n) &
                    \tup{v^m_1, \ldots v^m_p} = \exec{f(v_1, \ldots v_k \circ
                    v_x \ldots v_n)}_s 
                \end{array}
            }
            {
                \cI, \cO, \cE \vdash \Gamma[
                    x'_1 \rightarrow v'_1, \ldots x'_p \rightarrow v'_p],
                \sigma[x_k \rightarrow v_k] 
                \rightsquigarrow \\ \Gamma[x'_1 \rightarrow v'_1 \cdot v^m_1,
                \ldots x'_p \rightarrow v'_p \cdot v^m_p], 
                \sigma[x_k \rightarrow v_k \cdot v_x]
            }
            {Step}
   \end{array}
    $\end{center}
    \caption{Small Step Execution Semantics}
    \label{fig:exec_sem}
\end{figure}

\heading{Execution Semantics}
Figure~\ref{fig:exec_sem} presents the small step execution semantics for DDL.
Maps $\Gamma$ associates variable names to the data contained in the stream it represents.
Map $\sigma$ associates variable names to the data in the stream that has already been processed---representing the read-once semantics of \unix pipes.
Let $\tup{x'_1, \ldots x'_p} 
\leftarrow f(x_1, \ldots x_n)$ be a node in our DFG program. 
The function $\mathsf{choice}_f$ represents the order in which a commands consumes its inputs by returning a set of input indexes on which the function blocks on waiting to read.
For example, the $\mathsf{choice}_{cat}$ function for the command \ttt{cat} always returns the next non-closed index---as \ttt{cat} reads its inputs in sequence, each one until depletion.
\[
    \{k+1\} = \mathsf{choice}_{cat}(v_1 \cdot \perp, \ldots v_k \cdot \perp, v_{k+1}, \ldots v_n)
\]
For a $\mathsf{choice}_f$ function to be valid, it has to return an input index that has not been closed yet. Formally, 
\[
    \cS = \mathsf{choice}_f(v_1, \ldots v_k \cdot \perp, \ldots v_n) \implies k \notin
    \cS
\]
We assume that the set returned by $\mathsf{choice}_f$ cannot be \emph{empty} unless all input indexes are closed, meaning that all nodes consume all of their inputs until depletion even if they do not need the rest of it for processing.

The small step semantics nondeterministically picks a variable $x_k$, such that $k \in \mathsf{choice}_f(v_1, \ldots
v_n)$, \ie $f$ is waiting to read some input from $x_k$, and
$\sigma(x_k) < \Gamma(x_k)$, \ie there is data on the stream represented by
variable $x_k$ that has to be processed.
The execution then retrieves the next message $v_x$ to process,
and computes new messages $v^m_1, \ldots v^m_p$ to pass on to the output streams $x'_1, \ldots x'_p$.
Note that any of these messages (input or output) might be $\bot$. 
We pass $v_k \circ v_x$, which denotes that the previous data $v_k$ is now being combined with the new message $v_x$, to function $f$.
For all functions $f$ and new messages $v$, given
$
\tup{v'_1, v'_2, \ldots v'_p}
    = \exec{f(v_1, \ldots, v_k, \ldots v_n)}$ 
    we assume
the following constraint holds: 
\[
    \tup{v^m_1, \ldots v^m_p} = \exec{f(v_1, \ldots v_k \circ v_x, \ldots
    v_n)}_s 
\]
\[
   \iff
           \tup{v'_1 \cdot v^m_1, \ldots v'_p \cdot v^m_p}
    = \exec{f(v_1 \ldots, v_k \cdot v_x, \ldots v_n)}
\]
This constraint ensures that first processing arguments $v_1, \ldots v_k, \ldots
v_n$ and then message $v_x$ append to argument $k^{th}$ stream is equivalent to
processing messages $v_1, \ldots v_k \cdot v_x, \ldots v_n$ at once. Having this
property, allows our system to process messages as they arrive or wait for all
the messages to arrive, without changing the semantics of the execution.

The messages $v^m_1, \ldots v^m_p$ are passed on to their respective output
streams (by updating $\Gamma$). Note that the size of the output messages could vary, and they could even be empty.
Finally, $\sigma$ is updated to denote that $v_x$ has been processed.

\heading{Execution} 
Let $\tup{\cI, \cO, \cE}$ be a dataflow program,
where 
$\cI = \mathsf{input}~\overline{x}_i$ 
are the input variables, and $\mathsf{output}~\overline{x}_o$ are the output
variables.
Let $\sigma_i$ be the initial mapping from all variable names in 
the dataflow program $\tup{\cI, \cO, \cE}$ to 
empty string $\epsilon$. Let $\Gamma_i$ be the initial mapping for
variables in the dataflow program, such that all non-input variables $x \notin
\overline{x}_i$, map to the empty string $\Gamma_i(x) = \epsilon$. 
In contrast, all input variables $x \in \overline{x}_i$, \ie files already present in the file system, are mapped to the contents of the respective input file 
$\Gamma_i(x) = v \cdot \perp$.

\begin{figure}
    \begin{center}$    
    \begin{array}{c}
        \eat{
        \infral{
                \tup{x_1, x_2} \leftarrow \mathsf{tee}(x) \in \cE 
            }
        {\cI, \cO, \cE \vdash \Gamma[x \rightarrow v \cdot \bot, 
        x_1 \rightarrow v\cdot \bot, x_2 \rightarrow v
        \cdot \bot]}
        {Tee}
        \\
        \\
    }
        \infral{
                \begin{array}{cc}
                    \tup{x'_1, \ldots x'_p} \leftarrow f(x_1, \ldots x_n) \in
                    \cE&
                    \tup{v'_1 \cdot\bot, \ldots v'_p \cdot \bot} 
                    = \exec{f(v_1 \cdot \bot, \ldots v_n \cdot
                    \bot)} 
                \end{array}
            }
            {
                \cI, \cO, \cE \vdash \Gamma[x'_1 \rightarrow v'_1 \cdot \bot,
                \ldots x'_p \rightarrow v'_p \cdot \bot,
                x_1 \rightarrow v_1 \cdot \bot, \ldots x_n \rightarrow v_n \cdot
                \bot] 
            }
            {Completion}
        \end{array}
    $\end{center}
    \caption{Execution Constraints}
    \label{fig:exec_constraint}
\end{figure}

When no more small step transitions can take place (\ie all commands have finished processing), the dataflow execution terminates and the contents of output variables in $\cO$ can be written to their respective output files. 
Figure~\ref{fig:exec_constraint} represents the
constraint that has to be satisfied by $\Gamma$ at the end of execution, \ie when all variables are processed. We now prove some auxiliary theorems and lemmas to show that dataflow programs always terminate and that when they terminate, the constraint in Figure~\ref{fig:exec_constraint} holds.

\eat{
\begin{theorem}\label{thm:constraint}
    Let $\tup{x'_1, x'_2, \ldots, x'_p} \leftarrow f(x_1, \ldots x_n) \in \cE$,
    the following is true about $\Gamma$ and $\sigma$ at any point during the execution:
    \[
        \tup{\Gamma(x'_1), \ldots \Gamma(x'_p)} = \exec{f(\sigma(x_1), \ldots
        \sigma(x_n))}
    \]
\end{theorem}
\begin{proof}
  In appendix~\sx{thm:constraint-apx}.
\end{proof}

\begin{theorem}
    \label{thm:progress}
    Eventually for all variables $x$, $\exists~v. \Gamma(x) = v \cdot \perp$,
    \ie all variables will eventually be closed.
\end{theorem}
\begin{proof}
  In appendix~\sx{thm:progress-apx}.
\end{proof}
}
\eat{
\begin{theorem}
    The Dataflow program will always terminate. And when it terminates,
    constraint~\ref{fig:exec_constraint} will be true at the end state.
\end{theorem}
\begin{proof}
  In appendix~\sx{thm:termination-apx}.
\end{proof}
}

\begin{theorem}\label{thm:constraint-apx}
    Let $\tup{x'_1, x'_2, \ldots, x'_p} \leftarrow f(x_1, \ldots x_n) \in \cE$.
    During any point within the execution of the DFG, the following statement is
    true:
    \[
        \tup{\Gamma(x'_1), \ldots \Gamma(x'_p)} = \exec{f(\sigma(x_1), \ldots
        \sigma(x_n))}
    \]
\end{theorem}
\begin{proof}
    \noindent{Proof by induction on the number of execution steps.}\\
    \noindent{\it Base Case:}
    Let $\Gamma_i$ and $\sigma_i$ be the initial mappings.
   For $x \in \tup{x_1, \ldots x_n}$, $\sigma_i(x) =
    \epsilon$. For $x \in \tup{x'_1, \ldots x'_p}$, $\Gamma_i(x) = \epsilon$ 
    (since $x'_1, \ldots x'_p$ are not input variables to the DFG, they will be
    initialized to $\epsilon$).
    The following property is true for all functions $f$:
    \[
        \tup{\epsilon, \ldots \epsilon} = \exec{f(\epsilon, \ldots \epsilon)}
    \]
    Therefore, or the initial mappings, $\Gamma_i$ and $\sigma_i$,
    \[
        \tup{\Gamma_i(x'_1), \ldots \Gamma_i(x'_p)} = \exec{f(\sigma_i(x_1), \ldots
        \sigma_i(x_n))}
    \]

    \noindent{\it Induction Hypothesis:}
   Let $\Gamma$ and $\sigma$ be a snapshot of stream mappings during the
    execution of the DFG such that the following statement is
    true:
    \[
        \tup{\Gamma(x'_1), \ldots \Gamma(x'_p)} = \exec{f(\sigma(x_1), \ldots
        \sigma(x_n))}
    \]

    \noindent{\it Induction Case:}
    Let $\Gamma$ and $\sigma$ be a snapshot of stream mappings such that the
    induction hypothesis is true.
    Let $\Gamma'$ and $\sigma'$ be the snapshot after a single step of the
    execution takes place, given the snapshots $\Gamma$ and $\sigma$.

    If for all $i \in [1, n]$, $\sigma(x_i) = \sigma'(x_i)$, then a message
    updating $x'_1, \ldots x'_p$ was not processed (they can only be written by
    this node). Therefore, for all $k \in [1, p]. \Gamma(x'_p) = \Gamma'(x'_p)$
    and the following statement is true (assuming Induction Hypothesis)
    \[
        \tup{\Gamma'(x'_1), \ldots \Gamma'(x'_p)} = \exec{f(\sigma'(x_1), \ldots
        \sigma'(x_n))}
    \]

    If there exists an $i \in [1, n]$ such that $\sigma'(x_i) = \sigma(x_i)
    \cdot v_x$, $v_x \neq \epsilon$, then a message $v_x$ was processed. Note
    that the above statement can only be true for a single $i$. For all $k \neq
    i$, $\sigma(x_k) = \sigma'(x_k)$.

    The following statement is true from Induction Hypothesis:
    \[
        \tup{\Gamma(x'_1), \ldots \Gamma(x'_p)} = \exec{f(\sigma(x_1), \ldots
        \sigma(x_k), \ldots \sigma(x_n))}
    \]
    and from small step semantics, for all $k \in [1, p]$:
    \[
        \Gamma'(x_k) = \Gamma(x_k) \cdot v^m_k
    \]
    where:
    \[
        \tup{v^m_1, \ldots v^m_p} = \exec{f(\sigma(x_i), \ldots \sigma(x_k) \circ v_x, \ldots
    \sigma(x_n))}_s 
\]
    Using the definition of $\exec{\cdot}_s$, the following statement is true:
    \[
        \tup{\Gamma(x'_1) \cdot v^m_1, \ldots \Gamma(x'_p) \cdot v^m_p} =
        \exec{f(\sigma(x_1), \ldots
        \sigma(x_k) \cdot v_x, \ldots \sigma(x_n))}
    \]
Therefore, the following is true:
    \[
        \tup{\Gamma'(x'_1), \ldots \Gamma'(x'_p)} = \exec{f(\sigma'(x_1), \ldots
        \sigma'(x_k), \ldots \sigma'(x_n))}
    \] 
    The small step semantics will preserve this property.
    Therefore, by induction, 
    for all $\tup{x'_1, x'_2, \ldots, x'_p} \leftarrow f(x_1, \ldots x_n) \in \cE$,
    the following is always true about $\Gamma$ and $\sigma$, during any point
    within the execution:
    \[
        \tup{\Gamma(x'_1), \ldots \Gamma(x'_p)} = \exec{f(\sigma(x_1), \ldots
        \sigma(x_n))}
    \]
\end{proof}

\begin{lemma}
    \label{thm:progressf-apx}
    Let $\tup{x'_1, x'_2, \ldots, x'_p} \leftarrow f(x_1, \ldots x_n) \in \cE$.
    If $\forall i \in [1, n]$ $\Gamma(x_i)  = v_i \cdot \perp$, then eventually 
    $\forall i \in [1, n]$, $\sigma(x_i) = v_i \cdot \perp$ and eventually
    $\forall i \in [1, p]$, $\Gamma(x'_i) = v'_i \cdot \perp$.
\end{lemma}
\begin{proof}
    If for all $i \in [1, n]$, $\Gamma(x_i)$ is closed
    and $\mathsf{choice}_f$ is non empty unless $\sigma(x_i)$ is closed for all
    $i$, then eventually the
    execution will take a step to update $\sigma$ till, for all $i \in [1, n]$,
    $\sigma(x_i)$ is closed. 
    When all inputs are closed, $\exec{\cdot}$ dictates that all outputs will be
    closed as well.
    Using theorem~\ref{thm:constraint-apx}, $\Gamma(x'_i)$ will be closed.
\end{proof}

\begin{theorem}
    \label{thm:progress-apx}
    Eventually for all variables $x$, $\exists~v. \Gamma(x) = v \cdot \perp$,
    \ie all variables will eventually be closed.
\end{theorem}
\begin{proof}
    Let $\cC$ be the set of variables which will be closed eventually. 
    Note that $\cI \subseteq \cC$ (all input variables to the DFG will be
    eventually closed).
    Using Lemma~\ref{thm:progressf-apx},
    for any node 
$\tup{x'_1, x'_2, \ldots, x'_p} \leftarrow f(x_1, \ldots x_n) \in \cE$,
if $x_1, \ldots x_n \in \cC$, then $x'_1, \ldots x'_p \in \cC$.
Since the dataflow program contains no cycles, eventually all variables
    reachable from the input variables are in $\cC$. 
\end{proof}

\begin{theorem}
  \label{thm:termination-apx}
    The Dataflow program will always terminate. 
    Let $\Gamma$ and $\sigma$ be the stream mappings when the DFG terminates.
    The for $\Gamma$ and $\sigma$, the
    constraint~\ref{fig:exec_constraint} 
    will be true.
\end{theorem}
\begin{proof}
    The DFG graph terminates when all variables are closed.
    From Theorem \ref{thm:progress-apx}, all variables will eventually be closed.
    Constraint~\ref{fig:exec_constraint} follows from 
    Theorem \ref{thm:constraint-apx}, all variables being closed when DFG
    terminates, and the properties of $\exec{\cdot}$.
\end{proof}

\section{From Shell Scripts to DFGs and Back Again}
\label{translations}

\begin{figure}[t]
\centering
  \[
    \begin{array}{l rcl}
        \text{Commands} & c &:=& (s=w)^* w r^* \mid \text{pipe}_{\vert} \, c^+ \ttt{&}^? \mid c \, r^+  \\
        & &\mid&c \, \ttt{&} \mid ( c ) \mid c_1 \ttt{;} \, c_2 \mid c_1 \, \ttt{&&} \, c_2  \\
        & &\mid&c_1 \, \ttt{||} \, c_2 \mid \ttt{!} \, c \mid \ttt{while } \, c_1 \, c_2  \\
        & &\mid&\ttt{for } \, s \, w \, c \mid \ttt{if } \, c_1 \, c_2 \, c_3 \mid \ttt{case } \, w \, cb^*  \\
        & &\mid&s() \, c \\
        \text{Redirections} & r &:=& \ldots \\ 
        \text{Words} & w &:=& (s \mid \text{\textvisiblespace} \mid k)^* \\
        \text{Control Codes} & k &:=& \ldots \\
        \text{Strings} & s &\in& \Sigma^+
    \end{array}
  \] 
\caption{
  A relevant subset of shell syntax presented in Smoosh~\cite{smoosh:20}.
}
\label{fig:shell-semantics}
\end{figure}

This section formalizes the translations between the shell and our order-aware dataflow model.

\subsection{Shell to ODFM}
\label{shell2dfg}

Given a shell script the compiler starts by recursing on the AST, replacing subtrees in a bottom-up fashion with dataflow programs.
Fig.~\ref{fig:shell-semantics} shows a relevant subset of shell syntax, adapted from Smoosh~\cite{smoosh:20}.
Intuitively, some shell constructs (such as pipes \ttt{|}) allow for the composition of the dataflow programs of their components, while others (such as \ttt{;}) prevent it.
Figure~\ref{fig:comm_trans} shows the translation rules for some interesting constructs, and Figure~\ref{fig:aux-trans} shows several auxiliary relations that are part of this translation.
We denote compilation from a shell AST to a shell AST as $c \uparrow c'$, and compilation to a dataflow program as $c \uparrow \tup{p, b}$ where $p$ is a dataflow program and $b \in \{ \text{bg}, \text{fg} \} $ denotes whether the program is to be executed in the foreground or background .

The first two rules {\sc CommandTrans} and {\sc CommandId} describe compilation of commands. The bulk of the work is done in $\mathbf{cmd2node}$, which, when possible, defines a correspondence between a command and a dataflow node. Predicate $\mathbf{pure}$ indicates whether the command is pure, \ie whether it only interacts with its environment by reading and writing to files. All commands that we have seen until now (\ttt{grep, sort, uniq}) satisfy this predicate. The relations $\mathbf{ins}$ and $\mathbf{outs}$ define a correspondence between a commands arguments and the nodes inputs and outputs. We assume that a variable is uniquely identified from the file that it refers too, therefore if two variables have the same name, then they also refer to the same files.
Finally, relation $\mathbf{func}$ extracts information about the execution of the command (such as its \ttt{choice} function and \ttt{w}) to be able to reconstruct it later on.
\rev{
Note that the four relations $\mathbf{pure}$, $\mathbf{ins}$, $\mathbf{outs}$, and $\mathbf{func}$ act as axioms, and the soundness of our model and translations depends on their correctness.
Prior work~\cite{pash,posh} has shown how to construct such relations for specific commands using annotation frameworks, with PaSh providing annotations for more than 50 commands in POSIX and GNU Coreutils---two large and widely used sets of commands.
}

\begin{figure*}
    \begin{center}$    
    \begin{array}{cccc}
    \multicolumn{4}{c}{
      \infral{
        \begin{array}{cc}
          \text{cmd2node}(w, \overline{x}_o \leftarrow f(\overline{x}_i)) &
          \text{add\_metadata}(f, \overline{as}, \overline{r}) = f' \\
          \end{array}
          \\
          \text{redir}(\overline{x}_o, \overline{x}_i, \overline{r}, \overline{x}_o', \overline{x}_i')
      }
      {
        \overline{as} w \overline{r} \uparrow \tup{\text{input} \overline{x}_i' ; \text{output} \overline{x}_o' ; \overline{x}_o' \leftarrow f'(\overline{x}_i'), \text{fg}}
      }
      {CommandTrans}
    }
    \\
    \\
    \multicolumn{2}{c}{
      \infral{
          \text{cmd2node}(w, \bot)
      }
      {
        \overline{as} w \overline{r} \uparrow \overline{as} w \overline{r}
      }
        {CommandId}
    } 
    &
        \multicolumn{2}{c}{
    \infral
    {
        c \uparrow \tup{p, b}
    }
    {
      c \, \ttt{&} \uparrow \tup{p, \text{bg}}
    }
    {BackgroundDfg}
}
\\
\\
    \multicolumn{2}{c}{
    \infral
    {
        c \uparrow c'
    }
    {
      c \, \ttt{&} \uparrow c' \, \ttt{&}
    }
    {BackgroundCmd}
}
    &
    \multicolumn{2}{c}{
      \infral
      {
        \begin{array}{cc}
          c_1 \uparrow \tup{p_1, \text{bg}} &
          c_2 \uparrow \tup{p_2, b}
        \end{array}
      }
      {
        c_1 \, \ttt{;} \, c_2 \uparrow \tup{\text{compose}(p_1, p_2), b}
      }
      {SeqBothBg}
    }
   \\
   \\
    \multicolumn{2}{c}{
      \infral
      {
        \begin{array}{cc}
            c_1 \uparrow c_1' &
          c_2 \uparrow \tup{p_2, \text{bg}} 
        \end{array}\\
          \text{opt}(p_2) \Downarrow c_2'
      }
      {
        c_1 \, \ttt{;} \,  c_2 \uparrow c_1' ; (c_2' \, \ttt{&})
      }
      {SeqRightBg}
    }
    &
    \multicolumn{2}{c}{
      \infral
      {
        \begin{array}{cc}
          c_1 \uparrow \tup{p_1, \text{fg}} &
          c_2 \uparrow \tup{p_2, \text{bg}} \\
          \text{opt}(p_1) \Downarrow c_1' &
          \text{opt}(p_2) \Downarrow c_2'
        \end{array}
      }
      {
        c_1 \, \ttt{;} \, c_2 \uparrow c_1' ; (c_2' \, \ttt{&})
      }
      {SeqBothFgBg}
    }
    \\
    \\
    \multicolumn{2}{c}{
      \infral
      {
        \begin{array}{ccc}
          c_1 \uparrow c_1' &
          c_2 \uparrow \tup{p_2, \text{fg}}
        \end{array}
            \\
          \text{opt}(p_2) \Downarrow c_2'
      }
      {
        c_1 \, \ttt{;} \, c_2 \uparrow c_1' \, \ttt{;} \, c_2'
      }
      {SeqRightFg}
    }
    &
    \multicolumn{2}{c}{
      \infral
      {
        \begin{array}{cc}
          c_1 \uparrow c_1' &
          c_2 \uparrow c_2'
        \end{array}
      }
      {
        c_1 \, \ttt{;} \, c_2 \uparrow c_1' \, \ttt{;} \, c_2'
      }
      {SeqNone}
    }
    \\
    \\
    \multicolumn{4}{c}{
        \infral{
            \begin{array}{cc}
            c_1 \uparrow \tup{p_1, b_1}, \ldots c_n \uparrow \tup{p_n, b_n},
                &
            p'_1~ \ldots p'_{n-1} = \text{map}(\text{connectpipe}, p_1~ \ldots
            p_{n-1})
            \end{array}
                \\
            p = \text{fold\_left}(\text{compose}, p'_1~p'_2~\ldots p'_{n-1}~ p_n)
        }
        {pipe_{\vert} c_1~c_2 \ldots c_n~\& \uparrow \tup{p, bg} }
        {PipeBG}
    }
    \\
    \\
    \multicolumn{4}{c}{
        \infral{
            \begin{array}{cc}
            c_1 \uparrow \tup{p_1, b_1}, \ldots c_n \uparrow \tup{p_n, b_n},
                &
            p'_1~ \ldots p'_{n-1} = \text{map}(\text{connectpipe}, p_1~ \ldots
            p_{n-1})
            \end{array}
            \\
            p = \text{fold\_left}(\text{compose}, p'_1~p'_2~\ldots p'_{n-1}~ p_n)
        }
        {pipe_{\vert} c_1~c_2 \ldots c_n \uparrow \tup{p, fg} }
        {PipeFG}
    } 
\end{array}
  $\end{center}
  \caption{A subset of the compilation rules.}
  \label{fig:comm_trans}
\end{figure*}

\begin{figure*}
    \begin{center}$    
    \begin{array}{cc}
      \multicolumn{2}{c}{
        \infral{
          \begin{array}{cc}
            (\text{vars}(\cE_1) \setminus \text{vars}(\cI_1)) \cap (\text{vars}(\cE_2) \setminus \text{vars}(\cI_2)) = \emptyset &
            \cI' = \cI_1 \setminus \cO_2 \cup \cI_2 \setminus \cO_1 \\ 
            \cO' = \cO_1 \setminus \cI_2 \cup \cO_2 \setminus \cI_1 &
              p' = \cI'; \cO'; \cE_1 \cup \cE_2'
            \\
            \text{vars}(\cE_1) \setminus (\text{vars}(\cI_1) \cup \text{vars}(\cO_1) \cap \text{vars}(\cI_2) = \emptyset &
             \text{valid}(p')
              \\
 \multicolumn{2}{c}
            {
              \text{vars}(\cE_2) \setminus (\text{vars}(\cI_2) \cup \text{vars}(\cO_2) \cap \text{vars}(\cI_1) = \emptyset
          }
              \\
              \multicolumn{2}{c}{
              \vec{x}_1 = \text{vars}(\cE_1) \cap \text{vars}(\cE_2) \setminus  
              (\text{vars}(I_1) \cup \text{vars}(I_2) \cup \text{vars}(O_3) \cup
              \text{vars}(O_2)) }
          \\
    \end{array}
        \\
        \begin{array}{ccc}
              \vec{x}_2 \notin \text{vars}(\cE_1) \cup \text{vars}(\cE_2) 
            &
            \vert \vec{x}_1 \vert = \vert \vec{x}_2 \vert &
              \cE_2' = \cE_2[\vec{x}_2 / \vec{x}_1]
        \end{array}
             \\
                \forall x \in (\cI_1 \cup \cI_2 \setminus \cI') \cup (\cO_1
                \cup \cO_2 \setminus \cO'). \text{pipe}(x)
        }
        {
          \text{compose}(\cI_1;\cO_1;\cE_1, \cI_2;\cO_2;\cE_2) = p'
        }
        {}
      }
      \\
      \\
      \infral{
        \begin{array}{cccc}
          \text{pure}(w) &
          \text{ins}(w, \overline{x}_i, f) &
          \text{outs}(w, \overline{x}_o, f) &
          \text{func}(w, f)
        \end{array}
      }
      {
        \text{cmd2node}(w, \overline{x}_o \leftarrow f(\overline{x}_i))
      }
      {} &
      \infral{
        \neg \text{pure}(w)
      }
      {
        \text{cmd2node}(w, \bot)
      }
      {}
      \\
      \\
        \infral{
          \begin{array}{ccc}
            \text{ins}_f(\overline{x}_i, w, x_{in}?) &
            \text{outs}_f(\overline{x}_o, w, x_{out}?) &
            \text{func}(w, f)
          \end{array}
        }
        {
          \text{node2cmd}(\overline{x}_o \leftarrow f(\overline{x}_i), w, x_{in}?, x_{out}?)
        }
        {}
      &
      \infral{
          \cO' = \cO[x_{\text{stdin}}/x_{\text{stdout}}]
      }
      {\text{connectpipe}(\cI; \cO; \cE) = \cI;\cO';\cE}
      {}
      \\
      \\
      \multicolumn{2}{c}{
        \infral{
          \begin{array}{cc}
            \text{node2cmd}(\overline{x}_o' \leftarrow f(\overline{x}_i'), w, x_{in}?, x_{out}?) &
            \text{get\_metadata}(f) = \tup{\overline{as}, \overline{r}} \\
            \text{redir}(\overline{x}_o', \overline{x}_i', \overline{r}, \overline{x}_o, \overline{x}_i) &
            \overline{r}' = \text{in\_out}(\overline{r}, x_{in}?, x_{out}?)
          \end{array}
        }
        {
            \text{instantiate}(\overline{x}_o \leftarrow f(\overline{x}_i)) = \overline{as} w \overline{r} 
        }
        {}
      }
    \end{array}
    $\end{center}
    \caption{Auxiliary relations for translating commands to nodes and back.}
    \label{fig:aux-trans}
\end{figure*}

The rule {\sc BackgroundDfg} sets the background flag for the underlying dataflow program; if the operand of a \ttt{&} is not compiled to a dataflow program then it is simply left as is. The last part holds for all shell constructs, we currently only create dataflow nodes from a single command.

The next set of rules refer to the sequential composition operator ``\ttt{;}''. This operator acts as a dataflow barrier since it enforces an execution ordering between its two operands. Because of that, it forces the dataflow programs that are generated from its operands to be optimized (with opt) and then compiled back to shell scripts (with $\Downarrow$). However, there is one case ({\sc SeqBothBg}) where a dataflow region can propagate through a ``\ttt{;}'' and that is if the first component is to be executed in the background. In this case ``\ttt{;}'' does not enforce an execution order constraint between its two operands and the generated dataflow programs can be safely composed into a bigger one. The rules for ``\ttt{&&}'' and ``\ttt{||}'' are similar (omitted).

The relation $\mathbf{compose}$ unifies two dataflow programs by combining the inputs of one with the outputs of the other and vice versa. Before doing that, it ensures that the composed dataflow graph will be valid by checking that there is at most one reader and one writer for each internal and output variable, as well as all the rest of the dataflow program invariants, \eg the absence of cycles~\sx{model}. 

The remaining rules (not shown) introduce synchronization constraints and are not part of our parallelization effort---for example, we consider all branching operators as strict dataflow barriers.

\subsection{ODFM to Shell}
\label{dfg2shell}

\begin{figure*}
    \begin{center}$    
    \begin{array}{cc}
      \multicolumn{2}{c}{
        \infral{
          \begin{array}{ccc}
            \cI = \text{input } x_{i1}, \dots x_{in} &
 \cO = \text{output } x_{o1}, \dots x_{om} & 
            \cE = n_1, \dots n_k   
          \end{array}
              \\
              \text{\ttt{cin}} = \text{\ttt{cat} } file(x_{i1}) \text{\ttt{ > }} pipe(x_{i1}) \text{\ttt{ & ; }} \ldots \text{\ttt{ ; }} \text{\ttt{cat} } file(x_k) \text{\ttt{ > }} pipe(x_{in}) \text{\ttt{ &}} \\
           \text{\ttt{cin}} = \text{\ttt{cat} } pipe(x_{o1}) \text{\ttt{ > }} file(x_{o1}) \text{\ttt{ & ; }} \ldots \text{\ttt{ ; }} \text{\ttt{cat} } pipe(x_{om}) \text{\ttt{ > }} file(x_{om}) \text{\ttt{ &}} \\
           \text{\ttt{cnodes}} = \text{instantiate}(n_1) \text{\ttt{ & ; }} \ldots \text{\ttt{ ; }} \text{instantiate}(n_k) \text{\ttt{ & }}
        }
        {
          \text{body}(\cI, \cO, \cE) = \text{\ttt{cin ; cout ; cnodes}}
        }
        {}
      }
      \\
      \\
            \infral{
        \text{vars}(p) = \tup{v_1, \ldots, v_n}
      }
      {
        \text{prologue}(p) = \text{\ttt{mkfifo /tmp/p{1..n}}}
      }
      {} 
      &
      \infral{
        \begin{array}{cc}
        \text{prologue}(p) = pr & \text{epilogue}(p) = ep
        \end{array}
      }
      {
        p \Downarrow  pr \text{\ttt{ ; }} \text{body}(p) \text{\ttt{ ; }} ep
      }
      {}
 \\
        \\
        \multicolumn{2}{c}{
      \infral{
        \text{vars}(p) = \tup{v_1, \ldots, v_n}
      }
      {
        \text{epilogue}(p) = \text{\ttt{wait ; rm /tmp/p{1..n}}}
      }
      {}
  }
    \end{array}
    $\end{center}
    \caption{DFG to shell transformations.}
    \label{fig:dfg2shell-rules}
\end{figure*}

Figure~\ref{fig:dfg2shell-rules} presents the compilation $\Downarrow$ of a dataflow program $p = \cI; \cO; \overline{\cE}$ to a shell program.
The compilation can be separated in a prologue, the main body, and an epilogue.

The prologue creates a named pipe (\ie \unix FIFO) for every variable in the program.
Named pipes are created in a temporary directory using the \ttt{mkfifo} command, and are similar in behavior to ephemeral pipes except that they are explicitly associated to a file-system identifier---\ie they are a special file in the file-system.
Named pipes are used in place of ephemeral pipes (\ttt{|}) in the original script. 

The epilogue inserts a \ttt{wait} to ensure that all the nodes in the dataflow have completed execution, and then removes all named pipes from the temporary directory using \ttt{rm}.
The design of the prologue-epilogue pair mimics how \unix treats ephemeral pipes, which correspond to temporary identifiers in a hidden file-system.

The main body expresses the parallel computation and can also be separated into three components.
For each of the input variables $x_i \in \cI$, we add a command that copies the file $\text{\texttt{f}} = \text{file}(x_i)$ to its designated pipe.
Similarly, for all output variables $x_i \in \cO$ we add a command that copies the designated pipe to the output file in the filesystem $\text{\texttt{f}} = \text{file}(x_i)$.
Finally, we translate each node in $\cE$ to a shell command that reads from the pipes corresponding to its input variables and writes to the pipes corresponding to its output variables. In order to correctly translate a node back to a command, we use the node-command correspondence functions (similar to the ones for $\uparrow$) that were used for the translation of the command to a node. Since a translated command might get its input from (or send its output to) a named pipe, we need to also add those as new redirections with \emph{in\_out}. For example, for a node $x_3 \leftarrow f(x_1, x_2)$ that first reads $x_1$ and then reads $x_2$, where $f = \text{\texttt{grep -f}}$, the following command would be produced:
\medskip
\begin{minted}[fontsize=\footnotesize]{bash}
grep -f p1 p2 > p3 &
\end{minted}
\medskip

\section{Parallelization Transformations}
\label{transformations}

In this section we define a set of transformations that expose data parallelism on a dataflow graph. We start by defining a set of helper DFG nodes and a set of auxiliary transformations to simplify the graph and enable the parallelization transformations. Then we identify a property on dataflow nodes that indicates whether the node can be executed in a data parallel fashion. We then define the parallelization transformations and we conclude with a proof that applying all of the transformations preserves the semantics of the original DFG.

\subsection{Helper Nodes and Auxiliary Transformations}
\label{auxiliary}

Before we define the parallelization transformations, we introduce several helper functions that can be used as dataflow nodes. The first function is $\splitc$. $\splitc$ takes a single input variable (file or pipe) and sequentially splits into multiple output variables. The exact size of the data written in each output variable is left abstract since it does not affect correctness but only performance.
\[
    \overline{x} \leftarrow \splitc(x_i)
\]
\[
       \infral{
           \overline{v} = \tup{v_1 \cdot \perp, v_2 \cdot \perp, \ldots v_{k-1}
           \cdot \perp, v_k, \epsilon, \ldots \epsilon},\\
            v_c = v_1 \cdot v_2 \ldots \cdot v_k
        }
        {\forall~v_c.~\exec{\splitc(v_c)} = \overline{v}}
        {}
\]
\[
       \infral{
            \overline{v} = \tup{v_1 \cdot \bot, 
            v_2 \cdot \bot, \ldots v_m \cdot \bot}, ~~~~~
            v_c = v_1 \cdot v_2 \ldots \cdot v_m \cdot \bot
        }
        {\forall~v_c.~\exec{\splitc(v_c)} = \overline{v}}
        {}
\]
The second function is $\concat$, which coincidentally behaves the same as the \unix command \ttt{cat}.
$\concat$, given a list of input variables, combines their values and assigns it
to a single output variable. 
Formally $\concat$ is defined below:
\[
    x_i \leftarrow \concat(\overline{x})
\]
\[
        \infral{
            \overline{v} = \tup{v_1 \cdot \bot, v_2 \cdot \bot, 
            \ldots v_{k-1} \cdot \bot, v_k, \ldots v_m},\\
            v_c = v_1 \cdot v_2 \ldots \cdot v_k
        }
        {\forall~\overline{v}.~\exec{\concat(\overline{v})} = v_c}
        {}
\]
\[
        \infral{
            \overline{v} = \tup{v_1 \cdot \bot, v_2 \cdot \bot, 
            \ldots v_m \cdot \bot},
            v_c = v_1 \cdot v_2 \ldots \cdot v_m \cdot \bot
        }
        {\forall~\overline{v}.~\exec{\concat(\overline{v})} = v_c}
        {}
\]
The third function is $\mathsf{tee}$, which behaves the same way as the \unix command \ttt{tee}, i.e. copying its input variable to several output variables.
Formally $\mathsf{tee}$ is defined below:
\[
    \overline{x} \leftarrow \mathsf{tee}(x_i)
\]
\[
    \infral{
        \overline{v} = \tup{v_c, v_c, \ldots, v_c},
    }
    {\forall~v_c.~\exec{\mathsf{tee}(v_c)} = \overline{v}}
    {}
\]
\[
    \infral{
        \overline{v} = \tup{v_c \cdot \bot, v_c \cdot \bot, \ldots, v_c \cdot \bot},
    }
    {\forall~v_c.~\exec{\mathsf{tee}(v_c \cdot \bot)} = \overline{v}}
    {}
\]
The final function is $\relay$. $\relay$ works as an identity function. 
Formally $\relay$ is defined below:
\[
    x_i \leftarrow \relay(x_j),
    {\forall~v.~\exec{\relay(v)} = v}
    ,
    {\forall~v.~\exec{\relay(v \cdot \bot)} = v \cdot \bot}
\]

\begin{figure*}
    \begin{center}
        $
    \begin{array}{c}
     \infral{
         x_j \leftarrow \mathsf{Unused}(\cI, \cO, \cE),
        ~\cE' = \cE[x_j/x_i]
     }
        {\cI, \cO, \cE \Longleftrightarrow \cI, \cO, \cE' \cup \{ x_j \leftarrow
        \relay(x_i)\}}
        {Relay}
        \\
        \\
        \infral{
            x_s, x'_s \leftarrow \mathsf{Unused}(\cI, \cO, \cE),
         \\   E = \big\{ 
            \tup{x_s, x'_s} \leftarrow \splitc(x),
            \tup{x_1, \ldots x_k} \leftarrow \splitc(x_s),
            \tup{x_{k+1}, \ldots x_m} \leftarrow \splitc(x'_s)
            \big\}
        }
        {\cI, \cO, \cE \cup \big\{ \tup{x_1, \ldots x_m} \leftarrow \splitc(x) 
        \big\} \Longleftrightarrow \cI, \cO, \cE \cup
        E}
        {Split-Split}
        \\
        \\
        \infral{
            x_c, x'_c \leftarrow \mathsf{Unused}(\cI, \cO, \cE),~~~
         \\   E = \big\{ 
           x_c \leftarrow \concat(\tup{x_1, \ldots x_k}),
            x'_c \leftarrow \concat(\tup{x_{k+1}, \ldots x_m}),
            x \leftarrow \concat(\tup{x_c, x'_c})
            \big\}
        }
        {\cI, \cO, \cE \cup \big\{x \leftarrow \concat(\tup{x_1, \ldots x_m})
        \big\} \Longleftrightarrow \cI, \cO, \cE \cup
        E}
        {Concat-Concat}
        \\
        \\
        \infral{
            \overline{x} \leftarrow \mathsf{Unused}(\cI, \cO, \cE),~~~
            E = \big\{
                \overline{x} \leftarrow \splitc(x_j),
                x_i \leftarrow \concat(\overline{x})
            \big\}
        }
        {\cI, \cO, \cE \cup \big\{x_i \leftarrow \relay(x_j) \big\} 
        \Longleftrightarrow \cI, \cO, \cE \cup
        E}
        {Split-Concat}
        \\
        \\
        \infral{
            x^u_1, x^d_1, x^u_2, x^d_2, \ldots x^u_n, x^d_n \leftarrow
            \mathsf{Unused}(\cI, \cO, \cE),\\
            E = \big\{ 
            \tup{x^u_1, x^d_1} \leftarrow \mathsf{tee}(x_1), 
 \tup{x^u_2, x^d_2} \leftarrow \mathsf{tee}(x_2), 
        \ldots
 \tup{x^u_n, x^d_n} \leftarrow \mathsf{tee}(x_n),\\
        x_o \leftarrow \concat(x^u_1, x^u_2, \ldots x^u_n),
        x'_o \leftarrow \concat(x^d_1, x^d_2, \ldots x^d_n)
            \big\}
        }
        {\cI, \cO, \cE \cup \big\{x \leftarrow\concat({x_1, x_2, \ldots x_n}),
        \tup{x_o, x'_o} \leftarrow  \mathsf{tee}(x)\big\} 
        \Longleftrightarrow
        \cI, \cO, \cE \cup E
        }
        {Tee-Concat}
        \\
        \\
     \infral{
     }
        {\cI, \cO, \cE \cup \{ x_j \leftarrow \concat(x_i)\} 
        \Longleftrightarrow \cI', \cO, \cE' \cup \{ x_j \leftarrow
        \relay(x_i)\}}
        {One-Concat}
        \\
        \\
     \infral{
     }
        {\cI, \cO, \cE \cup \{ 
        x_j \leftarrow \splitc(x_i)\} 
        \Longleftrightarrow \cI', \cO, \cE' \cup \{ x_j \leftarrow
        \relay(x_i)\}}
        {One-Split}
\\
\\
\infral{
            x \leftarrow \mathsf{Unused}(\cI, \cO, \cE),\\
            \overline{x}_i = \tup{x_1, x_2, \ldots x_n},
            \overline{x}_j = \tup{x'_1, x'_2, \ldots x'_n},
          \\  E = \big\{ 
            x'_1 \leftarrow \relay(x_1),
            x'_2 \leftarrow \relay(x_2), 
            \ldots x'_n \leftarrow \relay(x_n)
            \big\}
        }
        {
            \cI, \cO, \cE \cup \big\{ 
            x \leftarrow \concat(\overline{x}_i),
            \overline{x}_j \leftarrow \splitc(x)
            \big\} \Longrightarrow \cI, \cO, \cE
            \cup E
        }
        {Concat-Split}
    \end{array}
$
\end{center}
\caption{Auxiliary Transformation}
\label{fig:transform_auxiliary}
\end{figure*}

\begin{figure*}
    \begin{center}
    $
    \infral{
        x_r, x_m, x^m_1, \ldots x^m_n, x^r_1, \ldots x^r_n, x^t_1, \ldots x^t_n,
            x^p_1, \ldots x^p_n\leftarrow \mathsf{Unused}(\cI, \cO, \cE) \quad 
            \text{dp}(f, f_m, f_r)\\
        E = \big\{ 
            \tup{x_1^c, \ldots, x_n^c, x_{n+1}^c } \leftarrow \mathsf{tee}(x^c) : \forall x^c \in \overline{x}
            \big\} \cup \\
            \big\{
            x_i^m \leftarrow f_m(x_i, \overline{x}_i) : \forall i \in \{ 1 \ldots n\}
            \big\} \cup
            \big\{
            x_r \leftarrow f_r(x_1^m, \ldots, x_n^m, \overline{x}_{n+1})
            \big\}
    }
            {\cI, \cO, \cE \cup 
                \big\{x \leftarrow \concat(x_1, \ldots, x_n), x_i \leftarrow f(x, \overline{x}) 
                \big\} \Longrightarrow \cI, \cO, \cE \cup E}
    {Parallel}
    $
    \end{center}
    \caption{Parallelization Transformation}
    \label{fig:transformation_parallel}
\end{figure*}

Using these helper nodes our compiler performs a set of auxiliary transformations that are
depicted in Figure~\ref{fig:transform_auxiliary}.
$\relay$ acts an identity function, therefore any edge can be transformed to
include a relay between them.
Spliting in multiple stages to get $n$ edges is the same as splitting in one
step into $n$ edges.
Similarly, combining $n$ edges in multiple stages is the same as combining $n$
edges in a single stage.
If we split an edge into $n$ edges, and then combine the $n$ edges back, this
behaves as an identity. A $\concat$ can be pushed following a $\mathsf{tee}$ by
creating $n$ copies of the $\mathsf{tee}$ function.
If a $\concat$ has single incoming edge, we can convert it into a relay. If a
$\splitc$ has a single outgoing edge, we can convert it into a relay.

The first seven transformations can be performed both ways.
The last transformations is one way.
A $\splitc$ after a $\concat$ can be converted into relays, if the input arity
of $\concat$ is the same as output arity of $\splitc$.
The reverse transformation in this case is not allowed as, using (Relay) rule, 
we can $\concat$ and $\splitc$ any two or more streams in the dataflow graph.
This will allow us to 
pass the output of any function in our graph to any other function as an input. 
This will break the semantics of our Dataflow graph.

\subsection{Data Parallelism and Transformations}
\label{ir:transformations}

The dataflow model exposes task parallelism as each different node can execute independently---only communicating with the other nodes through their communication channels.
In addition to that, it is possible to achieve data parallelism by executing some nodes in parallel by partitioning part of their input.

\heading{Sequential Consumption Nodes}
We are interested in nodes that produce a single output and consume their inputs in sequence (one after the other when they are depleted), after having consumed the rest of their inputs as an initialization and configuration phase. Note that there are several examples of shell commands that correspond to such nodes, e.g. \ttt{grep}, \ttt{sort}, \ttt{grep -f}, and \ttt{sha1sum}. Let such a node $x' = f(x_1, \ldots, x_{n+m})$, where w.l.o.g. $x_1, x_2, \ldots, x_n$ represent the configuration inputs and $x_{n+1}, \ldots, x_{n+m}$ represent the sequential consumption inputs. The consumption order of such a command is shown below:
\[
\mathsf{choice}_f(\overline{v}) = 
\begin{cases}
\{ i : i \leq n \wedge \neg closed(v_i) \} \text{ if} \neg \forall i \leq n, closed(v_i) \\
\{ i : \forall j < i, closed(v_j) \} \text{ otherwise}
\end{cases}
\]
If we know that a command $f$ satisfies the above property we can safely transform it to a $x^i = \concat(x_{n+1}, \ldots, x_{n+m})$ followed by a command $x' = f'(x^i, x_1, \ldots, x_n)$, without altering the semantics of the graph.

\heading{Data Parallel Nodes}
We now shift our focus to a subset of the sequential consumption nodes, namely those that can be executed in a data parallel fashion by splitting their inputs. These are nodes that can be broken down in a parallel \emph{map} $f_m$ and an associative \emph{aggregate} $f_r$. Formally, these nodes have to satisfy the following equation:
\begin{align*}
\forall n, &\llbracket f(v_1^i \cdots v_n^i, \overline{v}) \rrbracket = \llbracket f_r(\overline{v}_1, \ldots, \overline{v}_n, \overline{v}) \rrbracket \\
&\text{where } \forall i \in \{1 \ldots n \} \overline{v}_i = \llbracket f_m(v_i^i, \overline{v}) \rrbracket
\end{align*}
We denote data parallel nodes as $\text{dp}(f, f_m, f_r)$
Example of such a node that satisfies this property is the \ttt{sort} command, where $f_m = \text{\texttt{sort}}$ and $f_r = \text{\texttt{sort -m}}$.

In addition to the above equation, a \emph{map} function $f_m$ should not output anything when its input closes.
\[
    \llbracket f_m(v, \overline{v}) \rrbracket = \llbracket f_m(v \cdot \bot, \overline{v}) \rrbracket
\]
Note that $f_m$ could have multiple outputs and be different than the original function $f$. As has been noted in prior research~\cite{dq:17} this is important as some functions require auxiliary information in the map phase in order to be parallelized.
An important observation is that a subset of all data parallel nodes are completely stateless, meaning that $f_m = f$ and $f_r = \concat$, and therefore are embarrasingly parallel.

We can now define a transformation on any data parallel node $f$, that replaces it with a map followed by an aggregate. This transformation is formally shown in Figure~\ref{fig:transformation_parallel}. Essentially, all the sequential consumption inputs (that are concatenated using $\concat$) are given to different $f_m$ nodes the outputs of which are then aggregated using $f_r$ while preserving the input order. Note that the configuration inputs have to be duplicated using \ttt{tee} to ensure that all parallel $f_m$s and $f_r$s will be able to read them in case they are pipes and not files on disk.

Using the auxiliary transformations---by adding a \ttt{split} followed by \ttt{cat} before a data parallel node, we can always parallelize them using the parallelization transformation.

\heading{Correctness of Transformations}
We now proceed to prove a series of statements regarding the
semantics-preservation properties of dataflow programs.

\noindent{\it Program Equivalence:}
Let $p = \tup{\cI, \cO, \cE}$ and $p' = \tup{\cI', \cO', \cE'}$ be two dataflow programs,
where $\cI = \tup{x^i_1, \ldots x^i_n}$, $\cI' = \tup{y^i_1, \ldots y^i_n}$, 
$\cO = \tup{x^o_1, \ldots x^o_m}$, and $\cO' = \tup{y^o_1, \ldots y^o_m}$. These
programs are equivalent if and only if, assuming initial
value of input values are equal (for all $k \in [1, n].~x^i_k$ is equal to the
value of $y^i_k$), for 
the value of the output variables is the
same, when both of these DFGs terminate.
Formally, for all values $v_1, \ldots v_n$, 
if 
\[
    \forall~j \in [1, n].\Gamma_i(x^i_j) = \Gamma'_i(y^i_j) = v_j
\]
\[
    \forall~j \in [1, m].\Gamma_o(x^o_j) = \Gamma'_o(y^o_j)
\]
where $\Gamma_i$, $\Gamma'_i$ are the initial mappings
for $p$ and $p'$ respectively, and $\Gamma_o$, $\Gamma'_o$ are the mappings when $p$ and $p'$ have completed their execution.

\eat{
\begin{theorem}
    \label{thm:equiv_sub}
    Let $p = \tup{\cI, \cO, \cE \cup E}$ and
$p' = \tup{\cI, \cO, \cE \cup E'}$
    be two dataflow programs.
    Let $\cS_i$ be the set of input variables in node set $E$ (variables read in
    $E$ but not assigned inside $E$). Let $\cS_o$ be the set of output variables
    in the node set $E$ (variables assigned in $E$ but not read inside $E$).
    And $\cS_i$, $\cS_o$ be the input variables and output variables of $E'$ and
    $E$ as well. 
    If $\tup{\cS_i, \cS_o, E}$ is equivalent to $\tup{\cS_i, \cS_o, E'}$, then  
    program 
    $\tup{\cI, \cO, \cE \cup E'}$ is equivalent to $\tup{\cI, \cO, \cE \cup E}$.
\end{theorem}
\begin{proof}
  In appendix~\sx{thm:equiv_sub-apx}.
\end{proof}

\begin{theorem}
    Transformations presented in Figure~\ref{fig:transform_auxiliary} and Figure~\ref{fig:transformation_parallel} preserve program equivalence.
\end{theorem}
\begin{proof}
  In appendix~\sx{thm:correct-apx}.
\end{proof}
}

\begin{theorem}
    \label{thm:equiv_sub-apx}
    Let $p = \tup{\cI, \cO, \cE \cup E}$ and
$p' = \tup{\cI, \cO, \cE \cup E'}$
    be two dataflow programs.
    Let $\cS_i$ be the set of input variables in node set $E$ (variables read in
    $E$ but not assigned inside $E$). Let $\cS_o$ be the set of output variables
    in the node set $E$ (variables assigned in $E$ but not read inside $E$).
    Let $\cS'_i$, $\cS'_o$ be the input variables and output variables of $E'$.
    We assume $\cS_i = \cS'_i$ and $\cS_o = \cS'_o$.
    If $\tup{\cS_i, \cS_o, E}$ is equivalent to $\tup{\cS_i, \cS_o, E'}$, then  
    program 
    $\tup{\cI, \cO, \cE \cup E'}$ is equivalent to $\tup{\cI, \cO, \cE \cup E}$.
\end{theorem}
\begin{proof}
Given any initial mapping $\Gamma_i$, let $\Gamma_o$, $\Gamma'_o$ be the
    mappings when $p$ and $p'$ complete their execution. 
    For all $x \in \cS_i$, $\Gamma_o(x) = \Gamma'_o(x)$ as there are no cycles
    in the dataflow graph, and the subgraph which computes $\cS_i$ is same in
    both $p$ and $p'$. 

    \eat{
Since $\tup{\cS_i, \cS_o, E}$ is equivalent to $\tup{\cS_i, \cS_o, E'}$, 
    and since for all $x_k \in \cS_i$, $\Gamma_o(x_k) = \Gamma'_o(x_k)$,
    for all $x_j \in \cS_o$, $\Gamma_o(x_j) = \Gamma'_o(x_j)$.
}

    Since $\tup{\cS_i, \cS_o, E}$ is equivalent to $\tup{\cS_i, \cS_o, E'}$,
    and for all $x \in \cS_i$, $\Gamma_o(x) = \Gamma'_o(x)$,
    for all $x \in \cS_o$, $\Gamma_o(x) = \Gamma'_o(x)$.

    Only variables in set $\cS_o$ are the variables assigned in $E$ and $E'$, 
    that are used in computing the value of the output variables $\cO$. 
    Since the value of the variables is same in both these programs, given the
    same input mapping $\Gamma_i$, for all output variables $x \in \cO$,
    $\Gamma_o(x) = \Gamma'_o(x)$.

    Therefore, both these programs are equivalent.
\end{proof}

\begin{theorem}
    \label{thm:correct-apx}
    Transformations presented in Figure~\ref{fig:transform_auxiliary} and Figure~\ref{fig:transformation_parallel} preserve program equivalence.
\end{theorem}
\begin{proof}
    The (Relay) transformation preserves program
    equivalence as
    the program terminates, the value of $x_i$ is equal to the value of
    $x_j$. 

    The remaining transformations, transforming an input program 
$\tup{\cI, \cO, \cE \cup E}$ to an output program
    $\tup{\cI', \cO',\cE \cup E'}$.
    For all transformations $\cS_i = \cS'_i$ and $\cS_o = \cS'_o$ (where $\cS_i,
    \cS'_i, \cS_o, \cS'_o$ are defined above).
    First seven transformations, 
    equivalence of programs $\tup{\cS_i,
    \cS_o, E}$ and $\tup{\cS_i, \cS_o, E'}$
 follow from the execution semantics for
    $\concat$, $\relay$, $\splitc$, $\mathsf{tee}$, 
the properties
    of $f_m$ and $f_r$ for data parallel commands $f$.

   The (Concat-Split) transformation relies on the additional property that the
    program produces the same output independent of how the split breaks the
    input stream. Choice of a particular way of breaking the stream does not
    change the value of the program's output variables when it terminates.
    
    Since, $\tup{\cS_i,
    \cS_o, E}$ is equivalent to $\tup{\cS_i, \cS_o, E'}$, 
    these transformations preserve equivalence (Theorem~\ref{thm:equiv_sub-apx}).
\end{proof}

\section{Evaluation}
\label{eval}

Our evaluation consists of two parts. 
The first part is a case study of applying GNU Parallel to two scripts, demonstrating the difficulty of manually reasoning about parallel shell pipelines and the challenges that one has to address in order to achieve a parallel implementation.
The second part demonstrates the performance benefits of our transformations on 47 unmodified shell scripts.
Before discussing our evaluation, we offer a brief outline of the compiler implementation.

\heading{Implementation}
We reimplement the compilation and optimization phases of PaSh~\cite{pash} according to our model and associated transformations.
\rev{
The new implementation is about 1500 lines of Python code and uses the order-aware dataflow model as the centerpiece intermediate representation.
It is also more modular and facilitates the development of additional transformations, closely mirroring the back-and-forth shell-to-ODFM translations described in Section~\ref{translations} and the parallezing transformations described in Section~\ref{transformations}.
While we expect that most users would use PaSh by writing shell scripts, completely ignoring the ODFM, it is also possible to manually describe programs in the intermediate representation, enabling other frontend and backend frameworks to interface with it.
}

By reimplementing PaSh's optimization phase to mirror our transformations we also discovered and solved a bug in PaSh.
The old implementation did not \ttt{tee} the configuration inputs of a parallelized command, but rather allowed all parallel copies to read from the same input.
While this is correct if the configuration input is a file on disk, the semantics indicated that in the general case it leads to incorrect results---for example, in cases where this input is a stream---because all parallel commands consume items from a single stream, only reading a subset of them.

\subsection{Case Study: GNU Parallel}
\label{case-study}

We describe an attempt to achieve data parallelism in two scripts using GNU \ttt{parallel}~\cite{Tange2011a}, a tool for running shell commands in parallel.
We chose GNU \ttt{parallel} because it compares favorably to other alternatives in the literature~\cite{parallel-alts}, but note that GNU \ttt{parallel} sits somewhere between an automated compiler, like PaSh and POSH, and a fully manual approach---illustrating only some of the issues that one might face while manually trying to parallelize their shell scripts.

\heading{Spell}
We first apply \ttt{parallel} on \emph{Spell}'s first pipeline~\sx{intro}:
\medskip
\begin{minted}[fontsize=\footnotesize]{bash}
TEMP_C1="/tmp/{/}.out1"
TEMP1=$(seq -w 0 $(($JOBS - 1)) | sed 's+^+/tmp/in+' | sed 's/$/.out1/' | tr '\n' ' ')
TEMP1=$(echo $TEMP1)
mkfifo $TEMP1
parallel "cat {} | col -bx | tr -cs A-Za-z '\n' | tr A-Z a-z | \
    tr -d '[:punct:]' | sort > $TEMP_C1" ::: $IN &
sort -m $TEMP1 | parallel -k --jobs ${JOBS} --pipe --block "$BLOCK_SIZE" "uniq" | 
    uniq | parallel -k --jobs ${JOBS} --pipe --block "$BLOCK_SIZE" "grep -vx -f $dict -"
rm $TEMP1
\end{minted}
\medskip
\noindent
It took us a few iterations to get the parallel version right, leading to a few observations.
First, despite its automation benefits, \ttt{parallel} still requires manual placement of the intermediate FIFO pipes and $agg$ functions.
Additionally, achieving ideal performance requires some tweaking:
  setting \ttt{--block} to \ttt{10K}, \ttt{250K}, and \ttt{250M} yields widely different execution times---27, 4, and 3 minutes respectively.

Most importantly, omitting the \ttt{-k}  flag in the last two fragments breaks correctness due to re-ordering related to scheduling non-determinism.
These fragments are fortunate cases in which the \ttt{-k} flag has the desired effect, because their output order follows the same order as the arguments of the commands they parallelize.
Other commands face problems, in that the correct output order is not the argument order nor an arbitrary interleaving.

\heading{Set-difference}
We apply \ttt{parallel} to Set-diff, a script that compares two streams using \ttt{sort} and \ttt{comm}:
\medskip
\begin{minted}[fontsize=\footnotesize]{bash}
mkfifo s1 s2
TEMP_C1="/tmp/{/}.out1"
TEMP1=$(seq -w 0 $(($JOBS - 1)) | sed 's+^+/tmp/in+' | sed 's/$/.out1/' | tr '\n' ' ')
TEMP1=$(echo $TEMP1)
TEMP_C2="/tmp/{/}.out2"
TEMP2=$(seq -w 0 $(($JOBS - 1)) | sed 's+^+/tmp/in+' | sed 's/$/.out2/' | tr '\n' ' ')
TEMP2=$(echo $TEMP2)
mkfifo ${TEMP1} ${TEMP2}
parallel "cat {} | cut -d ' ' -f 1 | tr [:lower:] [:upper:] | sort > $TEMP_C1" ::: $IN &
sort -m ${TEMP1} > s1 &
parallel "cat {} | cut -d ' ' -f 1 | sort > $TEMP_C2" ::: $IN2 &
sort -m ${TEMP2} > s2 &
cat s1 | parallel -k --pipe --jobs ${JOBS} \
  --block "$BLOCK_SIZE" "grep -vx -f s2 -"
rm ${TEMP1} ${TEMP2}
rm s1 s2  
\end{minted}
\medskip
\noindent
In addition to the issues highlighted in Spell, this parallel implementation has a subtle bug.
GNU \ttt{parallel} spawns several instances of \ttt{grep -vx -f s2 -} that all read FIFO \ttt{s2}.
When the first parallel instance exits, the kernel sends a \ttt{SIGPIPE} signal to the second \ttt{sort -m}. 
This forces \ttt{sort} to exit, in turn leaving the rest of the parallel \ttt{grep -vx -f} instances blocked waiting for new input.

The most straightforward way we have found to address this bug is to remove (1) the ``\ttt{&}'' operator after the second \ttt{sort -m}, and (2) \ttt{s2} from \ttt{mkfifo}.
This modification sacrifices pipeline parallelism, as the first stage of the pipeline completes before executing \ttt{grep -vx -f}.
The \ttt{parallel} pipeline modified for correctness completes in 4m54s.
Our compiler does not sacrifice pipeline parallelism by using \ttt{tee} to replicate \ttt{s2} for all parallel instances of \ttt{grep -vx -f}~\sx{bg}, completing in 4m7s.

\subsection{Performance Results}

\heading{Methodology} 
We use three sets of benchmark programs from various sources, including GitHub, StackOverflow, and the \unix literature~\cite{bentley1986literate, bentley1985spelling, taylor2004wicked, bigrams, unix50sol, mcilroy1978unix}. 

\begin{itemize}
        \item {\bf Expert Pipelines:}
        The first set contains 9 pipelines:
            NFA-regex, Sort, Top-N, WF, Spell, Difference, Bi-grams, Set-Difference, and Shortest-Scripts.
        Pipelines in this set contain 2--7 stages (mean: 5.2), ranging from a scalable CPU-intensive \ttt{grep} stage in NFA-regex to a non-parallelizable \ttt{diff} stage in Difference.
        These scripts are written by \unix experts:
          a few pipelines are from \unix legends~\cite{bentley1986literate, bentley1985spelling, mcilroy1978unix}, one from a book on \unix scripting~\cite{taylor2004wicked}, and a few are from top Stackoverflow answers~\cite{bigrams}.

        \item {\bf Unix50 Pipelines:} 
        The second set contains 34 pipelines solving the \unix50 game~\cite{unix50}.
        This set is from a recent set of challenges celebrating of \unix's 50-year legacy, solvable by \unix pipelines.
        The problems were designed to highlight \unix's modular philosophy~\cite{mcilroy1978unix}.
        We found unofficial solutions to all-but-three problems on GitHub~\cite{unix50sol}, expressed as pipelines with 2--12 stages (mean: 5.58).
        They make extensive use of standard commands under a variety of flags, and appear to be written by non-experts---contrary to the previous set, they often use sub-optimal or non-\unix-y constructs.
        We execute each pipeline as-is, without any modification.

        \item {\bf COVID-19 Mass-Transit Analysis Pipelines:}
        The third set contains 4 pipelines that were used to analyze real telemetry data from bus schedules during the COVID-19 response in one of Europe's largest cities~\cite{oasa-article}.
        The pipelines compute several statistics on the transit system per day---such as average serving hours per day and average number of vehicles per day. %
        Pipelines range between 9 and 10 stages (mean: 9.2) and use typical \unix staples such as \ttt{sed}, \ttt{awk}, \ttt{sort}, and \ttt{uniq}.
\end{itemize}

\noindent
We use our implementation of PaSh to parallelize all of the pipelines in these benchmark sets, working with three configurations:

\begin{itemize}
  \item {\bf Baseline:}
  Our compiler simply executes the script using a standard shell (in our case \ttt{bash}) without performing any optimizations.
  This configuration is used as our baseline.
  Note that it is not completely sequential since the shell already achieves pipeline and task parallelism based on \ttt{|} and \ttt{&}. 
  \item {\bf No Cat-Split:} 
  Our compiler performs all transformations except {\sc Concat-Split}.
  This configuration achieves parallelism by splitting the input before each command and then merging it back.
  It is used as a baseline to measure the benefits achieved by the {\sc Concat-Split} transformation.
  \item {\bf Parallel:}
  Our compiler performs all transformations.
  The {\sc Concat-Split} transformation, which removes a $\concat$ with $n$ inputs followed by a $\text{split}$ with $n$ outputs, ensures that data is not merged unnecessarily between parallel stages.  
\end{itemize}

\noindent
Experiments were run on a 2.1GHz Intel Xeon E5-2683 with 512GB of memory and 64 physical  cores, Debian 4.9.144-3.1, GNU Coreutils 8.30-3, GNU Bash 5.0.3(1), OCaml 4.05.0, 
and Python 3.7.3.
All pipelines are set to (initially) read from and (finally) write to the file system.
For ``Expert Pipelines'', we use 10GB-collections of inputs from Project Gutenberg~\cite{hart1971project};
  for ``Unix50 Pipelines'', we gather their inputs from each level in the game~\cite{unix50} and multiply them up to 10GB.
For ``Bus Route Analysis Pipelines'' we use the real bus telemetry data for the year 2020 (\textasciitilde~3.4GB).
The inputs for all pipelines are split in 16 equal sized chunks, corresponding to the intended parallelism level.

\begin{figure*}[t]
  \centering
  \includegraphics[width=\textwidth]{\detokenize{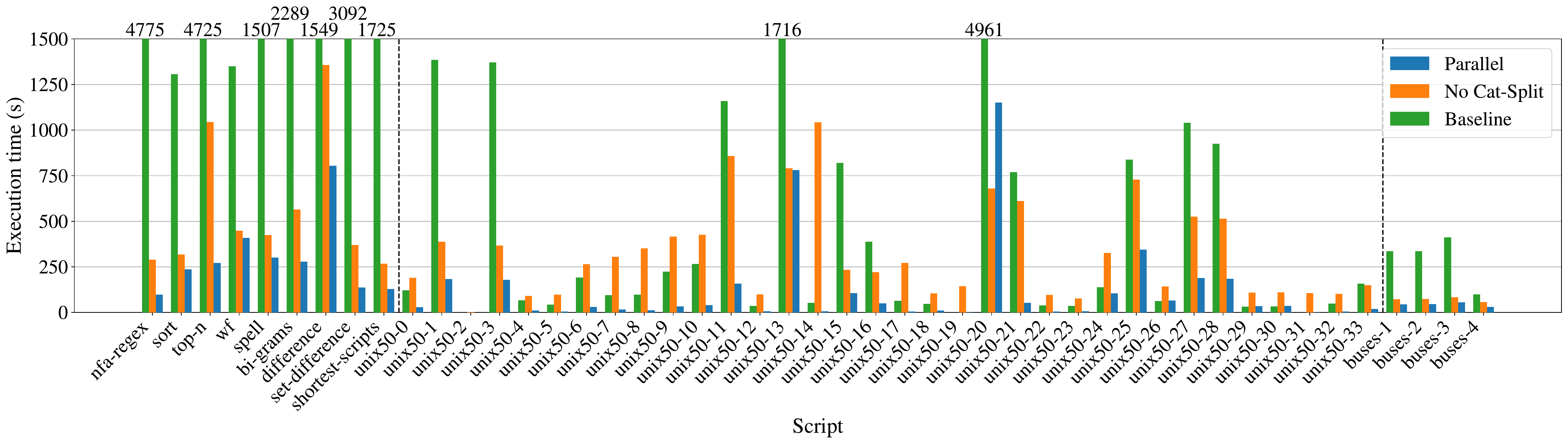}}
    \vspace{-3mm}
  \caption{
    Execution times with configurations \textbf{Baseline}, \textbf{Parallel}, and \textbf{No Cat-Split},
    with 16$\times$-parallelism.
  }
  \label{fig:oneliners-bar}
\end{figure*}

Fig.~\ref{fig:oneliners-bar} shows the execution times on all programs with $16\times$ parallelism and for all three configurations mentioned in the beginning of the evaluation.
It shows that all programs achieve significant improvements with the addition of the \textbf{Concat-Split} transformation.
The average speedup without \textbf{Concat-Split} over the \ttt{bash} baseline is 2.26$\times$.
The average speedup with the transformation is 6.16$\times$.

\pichskip{12pt}
\parpic[r][t]{
  \begin{minipage}{60mm}
  \includegraphics[width=\columnwidth]{\detokenize{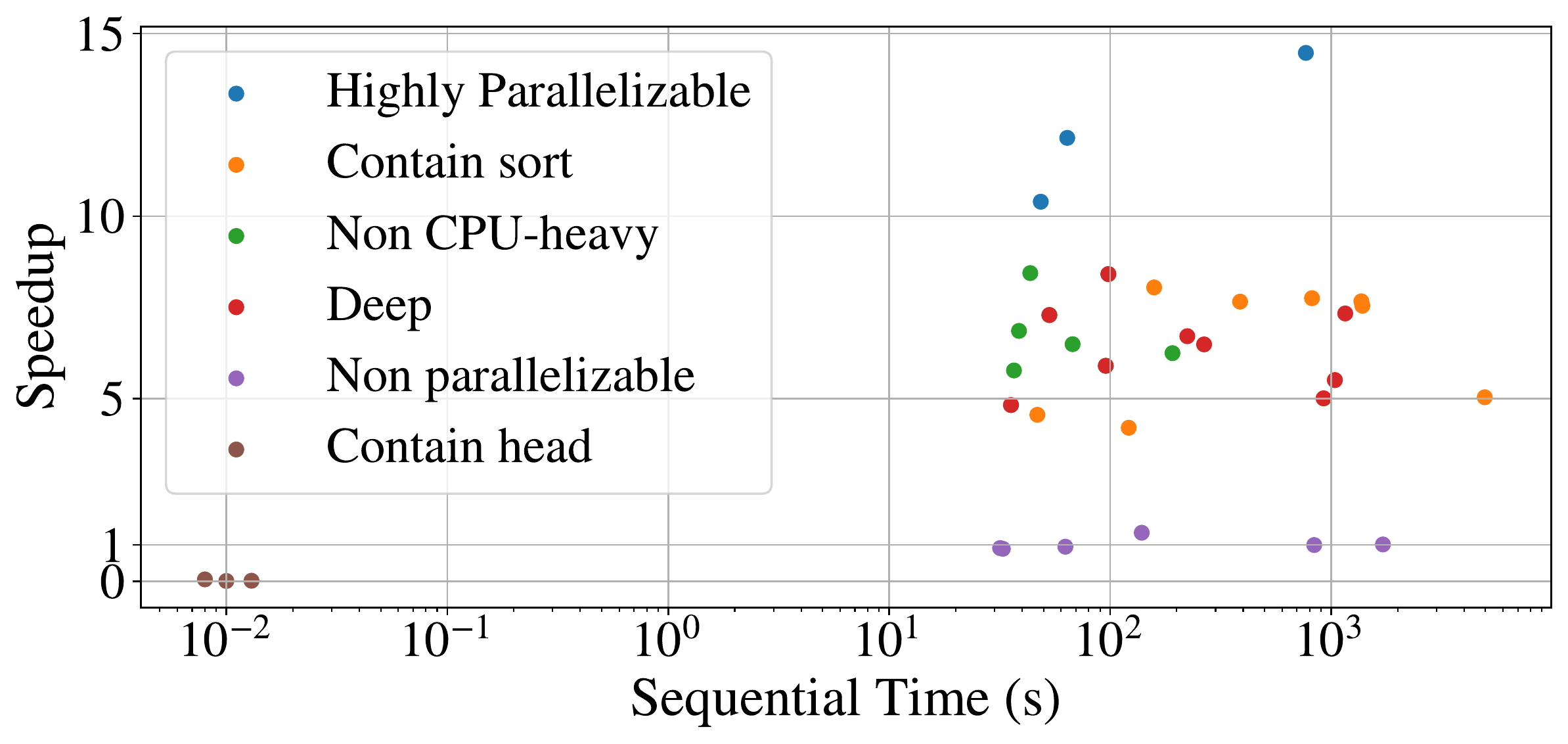}}
  \label{fig:scatter}
  \end{minipage}
}
The figure on the right explains the differences in the effect of the transformation based on the kind of commands involved in the pipelines.
It offers a correlation between sequential time and speedup, and shows that different programs that involve commands with similar characteristics (color) see similar speedups (y-axis).
Programs containing only parallelizable commands see the highest speedup (10.4--14.5$\times$).
Programs with limited speedup either 
  (1) contain \ttt{sort}, which does not scale linearly,
  (2) are not CPU-intensive, resulting in pronounced IO and constant costs, or
  (3) are deep pipelines, already exploiting significant pipeline-based parallelism.
Programs with non-parallelizable commands see no significant change in execution time (0.9--1.3$\times$).
Finally, programs containing \ttt{head} have a very small sequential execution, typically under $1s$, and thus their parallel equivalents see a slowdown due to constant costs---still remaining under $1s$.

\section{Related Work}
\label{related}

\heading{Dataflow Graph Models}
Graph models of computation where nodes represent units of computation and edges represent FIFO communication channels have been studied extensively~\cite{dennis1974first, karp1966properties, lee1987static, lee1987synchronous, kahn1974semantics, kahn1976coroutines}.
ODFM sits somewhere between Kahn Process Networks~\cite{kahn1974semantics, kahn1976coroutines} (KPN),
  the model of computation adopted by \unix pipes,
and Synchronous Dataflow~\cite{lee1987static, lee1987synchronous} (SDF).
A key difference between ODFM and SDF is that ODFM does not assume fixed item rates---a property used by SDF for efficient scheduling determined at compile-time.
Two differences between ODFM from KPNs is that (i) ODFM does not allow cycles, and (ii) ODFM exposes information about the input consumption order of each node.
This order provides enough information at compile time to perform parallelizing transformations while also enabling translation of the dataflow back to a \unix shell script.

Systems for batch~\cite{ mapreduce:08, murray2013naiad, spark:12}, 
  stream~\cite{ gordon2006exploiting, thies2002streamit, mamouras2017streamqre},
  and signal processing~\cite{lee1987static, bourke2013zelus} provide dataflow-based abstractions.
These abstractions are different from ODFM which operates on the \unix shell, an existing language with its own peculiarities that have guided the design of the model.

One technique for retrofitting order over unordered streaming primitives such as sharding and shuffling is to extend the types of elements using tagging~\cite{48862, watson1979prototype, arvind1984tagged}.
This technique would not work in the \unix shell, because (1) commands are black boxes operating on stream elements in unconstrained ways (but in known order), and (2) because data streams exchanged between commands contain flat strings, without support for additional metadata extensions, and thus no obvious way to augment elements with tags.
ODFM instead captures ordering on the edges of the dataflow graph, and leverages the consumption order of nodes (the choice function) in the graph to orchestrate execution appropriately.

Synchronous languages~\cite{signal, lustre, esterel, argos} model stream graphs as circuits where nodes are state machines and edges are wires that carry a single value.
Lustre~\cite{lustre} is based on a dataflow model that is similar to ours, but its focus is different as it is not intended for exploiting data-parallelism.

\heading{Semantics and Transformations}
Prior work proposes semantics for streaming extensions to relational query languages based on dataflow~\cite{li2005semantics,arasu2006cql}.
In contrast to our work, it focuses on transformations of time-varying relations.

More recently, there has been significant work on the correct
parallelization of distributed streaming applications by proposing
sound optimizations and compilation
techniques~\cite{hirzel2014catalog,schneider2013safe}, type systems~\cite{MSAIT2019}, and differential testing~\cite{kallas2020diffstream}. These efforts aim at producing a
parallel implementation of a dataflow streaming computation using techniques that do not require knowledge of the order of consumption of each node---a property that very important in our setting.

Recent work proposes a semantic framework for
stream processing that uses monoids to capture the type of data
streams~\cite{mamouras2020semantic}. That work mostly focuses on generality of expression,
showing that several already proposed programming models can be
expressed on top of it. It also touches upon
soundness proofs of optimizations using algebraic
reasoning, which is similar to our approach.

\heading{Divide and Conquer Decomposition}
Prior work has shown the possibility of decomposing programs or program fragments using divide-and-conquer techniques~\cite{dq:99, dq:17, dq:19, mrsynth:16}.
The majority of that work focuses on parallelizing special constructs---\eg loops, matrices, and arrays---rather than stream-oriented primitives.
Techniques for automated synthesis of MapReduce-style distributed programs~\cite{mrsynth:16} can be of significant aid for individual commands.
In some cases~\cite{dq:17, dq:19}, the map phase is augmented to maintain additional metadata used by the reducer phase.
These techniques complement our work, since they can be used to derive aggregators and the parallelizability properties of yet unknown shell commands, making them possible to capture in our model. 

\heading{Parallel Shell Scripting}
Tools exposing parallelism on modern \unix{}es such as \ttt{qsub}~\cite{gentzsch2001sun}, \textsc{SLURM}~\cite{yoo2003slurm}, AMFS~\cite{amfs} and \textsc{GNU} \ttt{parallel}~\cite{Tange2011a} are
predicated upon explicit and careful orchestration from their users.
Similarly, several shells~\cite{duff1990rc, mcdonald1988support, shell2, dagsh:17} add primitives for non-linear pipe topologies---some of which target parallelism.
Here too, however, users are expected to manually rewrite scripts to exploit these new primitives without jeopardizing correctness.

\rev{
Our work is inspired by PaSh~\cite{pash} and POSH~\cite{posh}, two recent systems that use command annotations to parallelize and distribute shell programs by operating on dataflow fragments of the Unix shell.
Our work is tied to PaSh~\cite{pash}, as it (i) uses its annotation framework for instantiating the correspondence of commands to dataflow nodes~\sx{shell2dfg}, and (ii) serves as its formal foundation since it reimplements all of its parallelizing transformations and proves them correct.
POSH~\cite{posh} too translates shell scripts to dataflow graphs and optimizes them achieving performance benefits, but its goal is to offload commands close to their data in a distributed environment.
Thus, POSH only performs limited parallelization transformations, and focusing more on the scheduling problem of determining where to execute each command.
It only parallelizes commands that require a concatenation combiner, \ie a subset of the transformations that we prove correct in this work, and thus replacing its intermediate representation with our ODFM would be possible.
POSH also proposes an annotation framework that captures several command characteristics.
Some of these characteristics, such as parallelizability, are also captured by PaSh.
Others are related to the scheduling problem---for example, whether a command such as \ttt{grep} produces output smaller than its input, making it a good candidate for offloading close to the input data.
}

\heading{POSIX Shell Semantics}
Our work depends on Smoosh, an effort focused on formalizing the semantics of the POSIX shell~\cite{smoosh:20}.
Smoosh focuses on POSIX semantics, whereas our work introduces a novel dataflow model in order to transform programs and prove the correctness of parallelization transformations on them.
One of the Smoosh authors has also argued for making concurrency explicit via shell constructs~\cite{smoosh:18}.
This is different from our work, since it focuses on the capabilities of the shell language as an orchestration language, and does not deal with the data parallelism of pipelines.

\heading{Parallel Userspace Environments}
By focusing on simplifying the development of distributed programs, a plethora of environments inadvertently assist in the construction of parallel software.
Such systems~\cite{ousterhout1988sprite, mullender1990amoeba, barak1998mosix}, languages~\cite{erlang:96, acute:05, mace:07}, or system-language hybrids~\cite{pike1990plan9, andromeda:15, cloudhaskell:11} hide many of the challenges of dealing with concurrency as long as developers leverage the provided abstractions---which are strongly coupled to the underlying operating or runtime system.
Even when these efforts are shell-oriented, such as  Plan9's \ttt{rc}, they are backward-incompatible with the \unix shell, and often focus primarily on hiding the existence of a network rather than on modelling data parallelism.

\rev{
\section{Discussion}
\label{discussion}

\heading{Command Annotations}
The translation of shell scripts to dataflow programs is based on knowledge about command characteristics, such as the predicate $\mathbf{pure}$ that determines whether a command does not perform any side-effect except for writing to a set of output files \sx{shell2dfg}.
In the current implementation, this information is acquired through the annotation language and annotations provided by PaSh~\cite{pash}.
An interesting avenue for future research would be to explore analyses for inferring or checking the annotations for commands. Such work could help extend the set of supported commands (which currently requires manual effort).
Furthermore, it would be interesting to explore extensions to the annotation language in order to enable additional optimizations; for example, commands that are commutative and associative could be parallelized more efficiently, by relaxing the requirement for input order and better utilizing the underlying resources.

\heading{Directly accessing the IR in the implementation}
As described earlier~\sx{eval}, our implementation currently allows manually developing programs in the ODFM intermediate representation. However, this interface is not that convenient to use as an end-user since it requires manually instantiating each node of the graph with the necessary command metadata, \eg inputs and outputs.
It would be interesting future work to design different frontends that interface with this IR. For example, a frontend compiler from the language proposed by \ttt{dgsh}~\cite{dagsh:17}; a shell that supports extended syntax for creating DAG pipelines. 
The IR could also act as an interface for different backends, for example one that implements ODFM in a distributed setting.

\heading{Parallel Script Debugging}
Debugging standard shell pipelines can be hard and it usually requires several iterations of trial and error until the user gets the script right. 
Our approach does not make the debugging experience any worse, as the system produces as output a parallel shell script, which can be inspected and modified like any standard shell script (as seen in \S\ref{informal}).
For example, a user could debug a script by removing a few stages of the parallel pipeline, or redirecting some intermediate outputs to permanent files for inspection.
This is possible because of the expressiveness of ODFM and the existence of a bidirectional transformation from dataflow programs to shell scripts, which allows the compiler to simply use a standard shell such as \ttt{bash} as its backend.

An approach that is particularly helpful, and which we have used ourselves, is to ask the compiler to add a $\mathit{relay}$ node between every two nodes of the graph and then instantiate this $\mathit{relay}$ node with an identity command that duplicates its input to its output and also a log file.
}

\medskip
\begin{minted}[fontsize=\footnotesize]{bash}
tee $LOG < $IN > $OUT 
\end{minted}
\medskip

\rev{
\noindent
This allows for stream introspection without affecting the behavior of the pipeline, facilitating debugging since the user can inspect all intermediate outputs at once.

\heading{Stream Finiteness and Extensions}
In our current model, parallelism is achieved by partitioning the finite stream, processing the partitions in parallel, and merging the results. 
Like PaSh and POSH, our model is designed to support terminating computations over large but finite data streams. 
All of the data processing scripts that we have encountered conform to this model and are terminating.
One way to extend our work to support and parallelize infinite streams---such as the ones produced by \ttt{yes} and \ttt{tail -f}---would involve repeated applications of partitioning, processing, and merging.
}

\section{Conclusion}
\label{conclusion}

We presented an order-aware dataflow model for exploiting data parallelism latent in \unix shell scripts.
The dataflow model is order-aware, accurately capturing the semantics of complex \unix pipelines:
  in our model, the order in which a node in the dataflow graph consumes inputs from different edges plays a central role in the semantics of the computation and therefore in the resulting parallelization.
The model captures the semantics of transformations that exploit data parallelism available in \unix shell computations and prove their correctness.
We additionally formalized the translations from the \unix shell to the dataflow model and from the dataflow model back to a parallel shell script.
We implemented our model and transformations as the compiler and optimization passes of PaSh, a system parallelizing shell pipelines, and used it to evaluate the speedup achieved on 47 data-processing pipelines.

While the shell has been mostly ignored by the research community for most of its 50-year lifespan, recent 
work~\cite{greenberg2021hotos-paper,greenberg2021hotos-panel,pash,posh,smoosh:20,spinellis2017extending} indicates renewed community interest in shell-related research.
We view our work 
  partly as providing the missing correctness piece of the shell-optimization work done by the systems community~\cite{pash,posh}, 
  and partly as a stepping stone for further studies on the dataflow fragment of the shell, \eg the development of more elaborate transformations and optimizations.

\begin{acks}
We thank 
  Konstantinos Mamouras for preliminary discussions that helped spark an interest for this work, 
  Dimitris Karnikis for help with the artifact,
  Diomidis Spinellis for benchmarks and discussions,
  Michael Greenberg and Jiasi Shen for comments on the presentation of our work,
  the anonymous ICFP reviewers and our shepherd Rishiyur Nikhil for extensive feedback, and 
  the ICFP artifact reviewers for their comments that significantly improved the paper artifact.
This research was funded in part by DARPA contracts {\sc HR00112020013} and {\sc HR001120C0191}, and NSF award {\sc CCF 1763514}.
Any opinions, findings, conclusions, or recommendations expressed in this material are those of the authors and do not necessarily reflect those of DARPA or other agencies.
\end{acks}

\bibliography{bib}

\end{document}